\documentclass[11pt]{amsart}
\usepackage{amsmath,amsfonts,amsthm,amscd,amssymb,graphicx}
\numberwithin{equation}{section} 

\newtheorem{theorem}{Theorem}[section]
\newtheorem{lemma}[theorem]{Lemma}
\newtheorem{proposition}[theorem]{Proposition}

\newtheorem{example}[theorem]{Example}
\newtheorem{remark}[theorem]{Remark}
\newtheorem{definition}[theorem]{Definition}
\newcommand{\sgn}{\text{\rm sgn }}
\newcommand{\Span}{\text{\rm Span }}
\newcommand{\unit}[1]{\ensuremath{\mathrm{#1}}}
\newcommand{\dif}{\ensuremath{\mathrm{d}}}

\def\bb1{{1\!\!1}}
\def\cF{ \mathcal{F}}

\def\trace{{\rm tr }}

\newcommand{\eps}{\ensuremath{\epsilon}}

\def\cW{\mathcal{W}}
\def\cV{\mathcal{V}}

\title{Spectral stability of ideal-gas shock layers}

\author[Humpherys, Lyng, and Zumbrun]{Jeffrey Humpherys, Gregory Lyng, and Kevin Zumbrun}

\thanks{
J.H. was partially supported by NSF grant DMS-0607721. 
G.L. was partially supported by a Basic Research Grant from the 
College of Arts and Sciences at the University of Wyoming.
K.Z. was partially supported  by NSF grant DMS-0300487.
}

\address{Department of Mathematics, Brigham Young University, Provo, UT 84602}
\email{jeffh@math.byu.edu}
\address{Department of Mathematics, University of Wyoming, Laramie, 
WY 82071}
\email{glyng@uwyo.edu}
\address{Department of Mathematics, Indiana University, Bloomington, IN 47402}
\email{kzumbrun@indiana.edu}

\date{Last Updated:  December 18, 2007}

\begin{document}

\begin{abstract}
Extending recent results in the isentropic case, 
we use a combination of asymptotic ODE estimates and numerical
Evans-function computations to examine the
spectral stability of shock-wave solutions of the compressible
Navier--Stokes equations with ideal gas equation of state.
Our main results are that, in appropriately rescaled
coordinates, the Evans function associated with the linearized
operator about the wave (i) converges in the large-amplitude limit
to the Evans function for a limiting shock profile of the
same equations, for which internal energy vanishes at one endstate;
and (ii) has no unstable (positive real part) zeros 
outside a uniform ball $|\lambda|\le \Lambda$.
Thus, the rescaled eigenvalue ODE for the set of all shock waves, 
augmented with the (nonphysical) limiting case, form a compact
family of boundary-value problems that can be conveniently 
investigated numerically.
An extensive numerical Evans-function study yields
one-dimensional spectral stability, independent of amplitude,
for gas constant $\gamma$ in  $[1.2, 3]$ and ratio $\nu/\mu$ 
of heat conduction to viscosity coefficient within 
$[0.2,5]$ ($\gamma\approx 1.4$, $\nu/\mu\approx 1.47$ for air).
Other values may be treated similarly but were not considered.
The method of analysis extends also to the multi-dimensional case,
a direction that we shall pursue in a future work.
\end{abstract}

\maketitle

\tableofcontents

\section{Introduction}\label{sec:intro}

A long-standing question in gas 
%JH
%and fluid 
dynamics is the 
stability of viscous shock layers, or traveling-wave solutions
$$
U(x,t)=\bar U(x-st), \quad
\lim_{z\to \pm \infty}\bar U(z)=U_\pm,
$$
of the compressible Navier--Stokes equations, where $U(x,t)=(v,u,E)^T$ 
is a vector recording specific volume, velocity, and total energy
of the 
%JH
%gas or 
fluid at location $x\in \mathbb{R}$ and time $t\in \mathbb{R}^+$.
A closely related question is the relation between Navier--Stokes
solutions and solutions of the formally limiting Euler equations in the
limit as viscosity and heat conduction coefficients go to zero:
more precisely, validity of formal matched asymptotics 
predicting that the Navier--Stokes solution consists approximately
of an Euler solution with smooth viscous shock layers replacing discontinuous
Euler shocks.

Recent progress in the form of
``Lyapunov-type'' theorems established in \cite{MZ.4,Z.3,GMWZ.1,GMWZ.2,GMWZ.3}
has reduced both problems to determination of spectral stability of
shock layers, i.e., the study of the eigenvalue ODE associated with
the linearized operator about the wave:
a standard analytically and numerically well-posed
(boundary value) problem in ODE that can be attacked
by the large body of techniques developed for asymptotic,
exact, and numerical study of ODE.
Indeed, the cited results hold for a substantially more general 
class of equations, and in one- or multi-dimensions.
%However, the determination of spectral stability is in general
%a difficult problem in systems of ODE.

In \cite{HuZ.1,PZ,FS,FS.2}, it has
been established in a similarly general context 
(general equations, one- and multi-dimensions),
using asymptotic ODE techniques, 
that spectral stability holds always in the small-amplitude limit, 
where amplitude is measured by $|U_+-U_-|$,
i.e., for shocks sufficiently close to a constant solution,
thus satisfactorily resolving the long-time stability and 
small-viscosity limit problems for small-variation solutions.

However, until very recently, the spectral stability of
{\it large-amplitude} shock waves has remained from a theoretical
viewpoint essentially open,
the sole exceptions being (i) a result of stability of Navier--Stokes shocks for
isentropic gas dynamics with $\gamma$-law gas in the 
special case $\gamma\to 1$, obtained by Matsumura-Nishihara \cite{MN}
quite early on using an ingenious energy estimate specific to that case; and (ii)
and a result of Zumbrun \cite{Z.4}---again obtained by energy
estimates special to the model---which establishes the stability of stationary phase-transitional shocks of
an isentropic viscous-capillary van der Waals model introduced by
Slemrod \cite{S}.

Progress instead has focused, quite successfully, on
the development of efficient
and general numerical methods for the study of stability of
individual waves, or compact families of waves, of essentially 
arbitrary systems; see, for example, \cite{Br.1,Br.2,BrZ,BDG,HuZ.2,BHRZ}.
These techniques, based on Evans-function computations,
effectively resolve the question of spectral stability for waves of 
large but finite amplitude, but leave open the question of stability
in the large-amplitude limit.
For discussion of the Evans function and its numerical computation, 
see \cite{AGJ,GZ,Z.3,BrZ,BHRZ} or Section \ref{evans} below.

Quite recently, however, Humpherys, Lafitte, and Zumbrun \cite{HLZ}
have introduced a new strategy combining asymptotic ODE techniques with 
numerical Evans-function computations,
by which they were able to carry out a global analysis
of shock stability in the isentropic $\gamma$-law case,
yielding stability independent of amplitude for 
$\gamma \in [1,2.5]$.\footnote{Other $\gamma$ may be treated 
similarly but were not 
%TODO: wording (but also grammar! ``carried out'' is incorrect)
%checked.}
%carried out.}
%tested.}
%examined.}
%examined in this study.}
considered.}
Specifically, after an appropriate rescaling, they showed by
a detailed asymptotic analysis of the linearized eigenvalue
ODE that the associated Evans functions (determining stability of the 
viscous shock profile) converge in the large-amplitude limit
to an Evans function associated with their formal limit, which
may then be studied either numerically or analytically (for example,
by energy estimate as in \cite{HLZ}).

The purpose of the present paper is to extend the approach
of \cite{HLZ} to the full (nonisentropic) Navier--Stokes equations
of compressible gas dynamics with ideal gas equation of state,
resolving in this fundamental case
the long-standing questions of viscous shock stability and
behavior in the small-viscosity limit.
Specifically, we show, as in the isentropic case, that
the Evans function indeed converges in the large-amplitude
limit, to a value corresponding to the Evans function of
a limiting system.
Compactifying the parameter range by adjoining this limiting system,
we then carry out systematic numerical Evans-function 
computations
as in \cite{BHRZ,HLZ} to determine stability for gas constant 
$\gamma \in [1.2,3]$
and (rescaled) ratio of heat conduction to viscosity coefficient
%TODO: update values.
%$\nu/\mu \in [.1,10]$.
$\nu/\mu \in [0.2,5]$, 
%in agreement with
well-including the
physical values
given in Appendices \ref{constants} and \ref{simple}.
The result, as in the isentropic case, is {\it unconditional stability,
independent of amplitude for an ideal gas equation of state.}
%TODO: include ratio $\nu=\kappa/c_v \mu$ of heat conduction to 
%viscosity?

%TODO: not sure of this part, probably leave out (KZ).
%(BUT, interesting to think about maybe...)
%Indeed, we find that the behavior of the Evans function 
%across the various parameter/argument regimes
%is largely determined by interpolation between
%its (analytically-established)
%large- and small-amplitude and high- and low-frequency limits.
%Near $\lambda=\infty$, $D(\lambda)\sim Ce^{\alpha \sqrt{\lambda}}$,
%giving characteristic shape of outer contour arc.
%Near $\lambda=0$, $D(\lambda)\sim c_0 + c_1 \lambda$, $c_j$ real, $c_0\ne 0$,
%giving tangency to imaginary axis near image of $\lambda=0$.
%However, typical ``loop'' seen in image of the imaginary axis is not
%explained by Taylor expansion!
%(unless it is $c_1\lambda + c_3\lambda^3$, with $c_2=0$..???)

\subsection{Discussion and open problems}\label{discussion}
% JH changed to present tense, minor wordcrafting
The asymptotic analysis of \cite{HLZ} is quite delicate; it depends sensitively both on the use of Lagrangian coordinates and on the precise way of writing the eigenvalue ODE as a first-order system.  It is thus not immediately clear that the analysis can be extended to the more complicated nonisentropic case.  Moreover, since Lagrangian coordinates---specifically, the associated
change of spatial variable 
\begin{equation}\label{change}
\dif\tilde x/\dif x=\rho (x),
\end{equation}
where $\rho$ is density---are not available in multi-dimensions, it is likewise, at first glance, unclear how how to extend the analysis beyond one spatial dimension.
%END JH

Remarkably, we find that the structure of the full, physical equations
is much more favorable to the analysis than that of the isentropic model.   
In particular, whereas in the isentropic case the eigenvalue
equations in the large-amplitude limit are a nonstandard 
singular perturbation of the limiting equations that must
be analyzed ``by hand'', in the full
(nonisentropic) gas case, they are a {\it regular perturbation}
for which convergence may be concluded by standard theorems
on continuous dependence of the Evans function with respect
to parameters; see, for example, the basic convergence lemma of \cite{PZ}.

%TODO This paragraph is cumbersome.
%gdl: Rewrote and slightly expanded this paragraph
Indeed, for $\gamma$ bounded away from the nonphysical case $\gamma=1$
(see Section \ref{prelim} for a description of the equations and the
physical background),
we have the striking difference that, for a fixed left endstate $U_-$, the
density remains uniformly bounded above and below
for all viscous shock profiles connecting $U_-$ to a right state $U_+$,
with energy going to infinity in the large-amplitude limit. By contrast, 
in the isentropic case, the density is artificially tied
to energy and thus density goes to infinity in the large-amplitude limit for
any $\gamma \ge 1$; see, e.g., \cite{Sm,Se.1,Se.2}.
This untangling of the large-amplitude behaviors of the density 
and the energy sets the stage for our analysis.  Below, 
to see this untangling, instead of fixing a left endstate $U_-$ and asking
which right endstates $U_+$ may be connected to $U_-$ by a viscous 
shock profile, we fix the shock speed $s=-1$ and all coordinates of the 
left state $U_-$ \emph{except} the energy. We find again 
that the density stays bounded above and below for all possible right
states $U_+$ connected by a shock profile to some such $U_-$.
%
%A related fact, relevant to the parametrization studied here,
%is that, fixing instead the shock speed $s=-1$
%and all coordinates of the left state $U_-$ 
%except energy, we find again
%that the density stays bounded above and below for all possible right
%states $U_+$ connected by a shock profile to some such $U_-$. In this case 
%the strong-shock limit corresponds to the point at which 
%energy goes to zero for the left-hand state $U_-$.
%END REWRITE-gdl

Since the equations remain regular so long as density is bounded from
zero and infinity, one important consequence of this fact is that
we need only check a few basic properties such as uniform decay of profiles
and continuous extension of stable/unstable subspaces to conclude
that the strong-shock limit is in the nonisentropic case
a regular perturbation of the limiting system as claimed;
see Section \ref{formulation} for details.

%This fact has the important implication that Lagrangian and
A second important consequence is that Lagrangian and
Eulerian coordinates are essentially equivalent in the nonisentropic
case so long as $\gamma$ remains uniformly bounded from $1$,
whereas, in the isentropic case, the equations become singular
for Eulerian ($x$) coordinates in the large-amplitude limit, 
by \eqref{change} together with the fact that density goes to infinity.
Here, we have chosen to work with Lagrangian coordinates
for comparison with previous one-dimensional analyses 
in the isentropic case \cite{BHRZ,HLZ,CHNZ}.
However, we could just as well have worked in Eulerian coordinates, 
including full multi-dimensional
effects, to obtain by the same 
arguments that the large-amplitude limit is a regular perturbation 
of the (unintegrated) limiting eigenvalue equation, and therefore the Evans
functions converge in the limit also in this 
multi-dimensional setting.

Likewise, uniform bounds on unstable eigenvalues may be obtained
in multi-dimensions by adapting the
asymptotic analysis of \cite{GMWZ.1} similarly as we have adapted
in Section \ref{high} the asymptotic analysis of \cite{MZ.3}.
Thus, for $\gamma$ uniformly bounded from $1$, 
{\it the analysis of this paper extends 
%CHANGED: don't want to depreciate multi-d...
%in routine fashion 
with suitable modification
to the multi-dimensional case}, making possible the resolution of 
multi-dimensional viscous stability by a systematic numerical 
Evans-function study as in the one-dimensional case.  
We shall carry out 
%CHANGED: likewise, don't depreciate multi-d contribution...
%these computations 
the multi-dimensional analysis
in a following work \cite{HLyZ}.

Presumably, the same procedure of compactifying the parameter
space after rescaling to bounded domain would work for any
gas law with appropriate asymptotic behavior as $\rho, e\to \infty$.
Thus, we could in principle investigate also van der Waals gas/fluids,
for example, which could yield interesting different behavior: in particular,
(as known already from stability index considerations \cite{ZS,Z.3}) 
instability in some regimes.
Other interesting areas for investigation include the study of
boundary layer stability (see \cite{CHNZ} for an analysis 
of the isentropic case), and stability of weak and strong detonation
solutions (analogous to shock waves) of the compressible Navier--Stokes
equations for a reacting gas.
A further interesting direction is to investigate the effects 
of temperature dependence of viscosity and heat conduction
on behavior for large amplitudes; see 
%TODO: restore?
%Section \ref{tdepcase} and 
Appendices \ref{kinetic} and \ref{tempdep}.
%TODO: update cross-ref., statement? 

In this work, we have restricted to the parameter range 
$\gamma \in [1.2,3]$ and 
%$\nu/\mu\in [.1,10]$, where
%$\nu/\mu\in [1,2]$, where
$\nu/\mu\in [0.2,5]$, where
$\gamma$ is the gas constant,
$\nu=\kappa/c_v$ is a rescaled coefficient of heat conduction ($\kappa$ the
Fourier conduction and $c_v$ specific heat),
and $\mu$ is the coefficient of viscosity; see 
equations. \eqref{eq:mass}--\eqref{eq:energy}, Section \ref{prelim}.
Similar computations may be carried out for arbitrary $\gamma$
bounded away from the nonphysical limit $\gamma=1$.  To approach
the singular limit $\gamma=1$ would presumably require a nonstandard
singular perturbation analysis like that of \cite{HLZ} in the isentropic case,
as the structure is similar; see Remark \ref{eblowup}.
The limits $\nu/\mu \to 0$ and $\nu/\mu \to \infty$ are more standard
singular perturbations with fast/slow structure that should be treatable
by the methods of \cite{AGJ}; this would be a very interesting direction
for further study.
%TODO: put back in the below sentence? (if numerics done...)
We note that our results for large 
%and small 
$\nu/\mu$ 
%indeed indicate
do indicate possible
further simplification in behavior, as the singular perturbation structure
would suggest; see Remark \ref{smallconj} and Figure \ref{fig0d}.
%ENDTODO
%NOTE: small-amplitude too could be done by numerics/singular pert.,
%though no need... might be useful in more general situations, though
%(e.g., nonclassical waves- see [Brin]).
For dry air at normal temperatures, $\gamma\approx 1.4$ and
$\nu/\mu\approx 1.47$, well within range; see Appendix \ref{constants}.

Finally, we mention the issue of rigorous verification.
Our results, though based on rigorous analysis, do not constitute
numerical proof, and are not intended to.
In particular, we do not use interval arithmetic.
%For example, in some places we have substituted numerically
%efficient but nonrigorous convergence studies
%based on abstract convergence results for computable but
%inefficient rigorous A Priori bounds.
%TODO: did we really do the above? In the end, for main
%compuational regime \nu\in [1/2,2], we did not I guess,
%or at least did only in nonnecessary ways (esp. setting $L$
%value... also, lambda step-size, etc.
%MAYBE SUBSTITUTE a less strong statement here?
%(a posteriori vs. a priori?)
%NO, BASICALLY OK I GUESS.... -K
Nonetheless, the numerical evidence for stability appears 
overwhelming, particularly in view of the fact that the family
$\{D(\lambda,v_+)\}$ of Evans contours estimated in the 
stability computations is analytic in both parameters,
yielding extremely strong interpolation estimates by the rigidity
of analytic functions.

In any case, our analysis contains all of the elements necessary
for numerical proof, the effective realization of which, however,
is a separate problem of independent interest.
Given the fundamental nature of the problem, we
view this as an important area for further investigation.

\subsection{Plan of the paper}
In Section \ref{prelim}, we set up
the problem, describing the equations, rescaling appropriately,
and verifying existence and uniform decay of profiles independent
of shock strength.
In Section \ref{formulation}, we construct the Evans function
and establish the key fact that it is continuous in all parameters
up to the strong-shock limit.
In Section \ref{high}, we carry out the main technical work of the
paper, establishing an upper bound on the modulus of unstable 
eigenvalues of the linearized operator about the wave in terms
of numerically approximable quantities associated with the traveling-wave
profile.
In Section \ref{numprot}, we describe our numerical 
method, first estimating a maximal radius within which
unstable eigenvalues are confined, then computing the winding 
number of the Evans function around the semicircle with that
radius to estimate the number of unstable eigenvalues, for
(a discretization of)
all parameters within the compact parameter domain, including
the strong-shock limit.
Finally, in Section \ref{numresults}, we perform the numerical
computations indicating stability.

In Appendices \ref{constants} and \ref{simple}, we discuss further the 
dimensionless constants $\Gamma$ and $\nu/\mu$,
and determine their values for air and other common gases.
In Appendix \ref{dense}, we discuss equations of state for 
%more general materials.
fluids and dense gases.
In Appendix \ref{mach}, we compute a formula for the Mach number,
a useful dimensionless quantity measuring shock strength independent
of scaling. 
%Finally, in Appendix \ref{liftbd}, we give a general
In Appendix \ref{liftbd}, we give a general
bound on the operator norm of lifted matrices acting on exterior products,
useful for analysis of the exterior product method of \cite{Br.1,BDG}.
In Appendix \ref{tempdep}, we discuss the 
changes needed to accommodate 
temperature-dependence in the coefficients of
viscosity and heat conduction, as predicted by the 
kinetic theory of gases.
\bigbreak

\section{Preliminaries}\label{prelim}

In Lagrangian coordinates, the Navier--Stokes equations for  
compressible gas dynamics take the form
\begin{align}
v_t-u_x&=0,\label{eq:mass}\\
u_t+p_x&=\left(\frac{\mu u_x}{v}\right)_x,\label{eq:momentum}\\
E_t+(pu)_x&=\left(\frac{\mu uu_x}{v}\right)_x+\left(\frac{\kappa T_x} 
{v}\right)_x,\label{eq:energy}
\end{align}
where $v$ is the specific volume, $u$ is the velocity, $p$ is the  
pressure, and the energy $E$ is made up of the internal energy $e$  
and the kinetic energy:
\begin{equation}
E=e+\frac{u^2}{2}.\label{eq:internal_kinetic}
\end{equation}
The constants $\mu$ and $\kappa$ represent viscosity and heat  
conductivity. Finally, $T$ is the temperature, and we assume that the  
internal energy $e$ and the pressure $p$ are known functions of the  
specific volume and the temperature:
\begin{equation*}
p=p_0(v,T),\quad e=e_0(v,T).
\end{equation*}

An important special case occurs when we consider an ideal,  
polytropic gas. In this case the energy and pressure functions take  
the specific form
\begin{equation}
p_0(v,T)=\frac{\bar RT}{v},\quad e_0(v,T)=c_vT,
\label{eq:ideal_gas}
\end{equation}
where $\bar R> 0$ and $c_v>0$ are constants that characterize the gas.
Alternatively, the pressure may be written as
\begin{equation}\label{Gammaeq}
p=\frac{\Gamma e}{v},
\end{equation}
where $\Gamma =\gamma -1= \frac{\bar R}{c_v }> 0$,
$\gamma > 1$ the 
%gas constant.
adiabatic index.
Equivalently, in terms of the entropy and specific volume, the pressure reads
$$
p(v,S)=ae^{S/c_v}v^{-\gamma},
$$
where $S$ is thermodynamical entropy,
or $p(v)=av^{-\gamma}$ in the isentropic approximation; see \cite 
{Sm,BHRZ,HLZ}.

In the thermodynamical rarefied gas approximation, $\gamma>1$ is the 
average over constituent particles of $\gamma=(N+2)/N$, where $N$ is the 
number of internal degrees of freedom of an individual particle,
or, for molecules with ``tree'' (as opposed to ring, or other
more complicated) structure, 
\begin{equation}\label{gammaformula}
\gamma=\frac{2n+3}{2n + 1}, 
\end{equation}
where $n$ is the number of constituent atoms \cite{Ba}: 
$\gamma= 5/3 \approx 1.66$ for monatomic, $\gamma= 7/5=1.4$ for diatomic gas.  
%For dense fluids, $\gamma$ is typically determined phenomenologically \cite{H}.
For fluids or dense gases, $\gamma$ is typically determined 
phenomenologically \cite{H}.
In general, $\gamma$ is usually taken within $1 \leq \gamma \leq 3$ 
in models of 
%gas- or fluid-dynamical flow, 
gas-dynamical flow, 
whether phenomenological or derived by statistical 
mechanics \cite{Sm,Se.1,Se.2}.
%TODO: not exactly true- take $\gamma=4-7$ for fluid (water, mercury) modeling
%in ``stiffened'' equation of state; what IS true (maybe) is that for situations
%where ideal gas eos is reasonable, $\gamma$ seems to be in this range...
%(I'm not even completely positive about this weaker statement,
%but think it's true..-K) DONE.

\subsection{Viscous shock profiles}\label{sec:vp}
A \emph{viscous shock profile} of \eqref{eq:mass}--\eqref{eq:energy}  
is a traveling-wave solution,
\begin{equation}
v(x,t)=\hat v(x-st),\;\;u(x,t)=\hat u(x-st),\;\; T(x,t)=\hat T(x-st),
\label{eq:tw_ansatz}
\end{equation}
moving with speed $s$ and connecting constant states $(v_\pm,u_\pm,T_ 
\pm)$. Such a solution is a stationary solution of the system of PDEs
\begin{equation}
v_t-sv_x-u_x=0,
\label{eq:moving_mass}
\end{equation}
\begin{equation}
u_t-su_x+p_0(v,T)_x=\left(\frac{\mu u_x}{v}\right)_x,
\label{eq:moving_momentum}
\end{equation}
\begin{equation}
\big[e_0(v,T)+u^2/2\big]_t
-s\big[e_0(v,T)+u^2/2\big]_x
+(p_0(v,T)u)_x
=\left(\frac{\mu uu_x}{v}\right)_x+\left(\frac{\kappa T_x}{v}\right)_x.
\label{eq:moving_energy}
\end{equation}

\subsection{Rescaled equations}\label{sec:rescaled}
Under the rescaling
\begin{equation}\label{scaling}
(x,t,v,u,T)\to (-\epsilon sx, \epsilon s^2t, v/\epsilon, -u/(\epsilon  
s), T/(\epsilon^2 s^2)),
\end{equation}
where $\epsilon$ is chosen so that $v_-=1$,
the system \eqref{eq:moving_mass}-\eqref{eq:moving_energy} becomes
\begin{align}
v_t+v_x-u_x&=0,\label{eq:rescale_mass}\\
u_t+u_x+p_x&=\left(\frac{\mu u_x}{v}\right)_x,\label 
{eq:rescale_momentum}\\
E_t+E_x+(pu)_x&=\left(\frac{\mu uu_x}{v}\right)_x+\left(\frac{\kappa  
T_x}{v}\right)_x,\label{eq:rescale_energy}
\end{align}
where the pressure and internal energy in the (new) rescaled  
variables are given by
\begin{equation}
p(v,T)=\epsilon^{-1}s^{-2}p_0(\epsilon v,\epsilon^2s^2T)
\label{eq:new_pressure}
\end{equation}
and
\begin{equation}
e(v,T)=\epsilon^{-2}s^{-2}e_0(\epsilon v,\epsilon^2s^2 T);
\label{eq:new_energy}
\end{equation}
in the ideal gas case, {\it the pressure and internal energy laws
remain unchanged}
\begin{equation}
p(v,T)=\frac{\bar R T}{v},
\qquad
e(v,T)=c_vT,
\label{eq:ideal_gas_scaling1}
\end{equation}
with the same $\bar R$, $c_v$.  Likewise, $\Gamma$ remains unchanged in \eqref 
{Gammaeq}.

\subsection{Rescaled profile equations}\label{sec:viscous}

Viscous shock profiles of 
\eqref{eq:rescale_mass}--\eqref{eq:rescale_energy}
satisfy the system of ordinary  differential equations
\begin{align}
v'-u'&=0,\label{eq:ode_mass}\\
u'+p(v,T)'&=\left(\frac{\mu u'}{v}\right)',\label{eq:ode_momentum}\\
\big[e(v,T)+u^2/2\big]'+(p(v,T)u)'&=\left(\frac{\mu uu'}{v}\right)'+ 
\left(\frac{\kappa T'}{v}\right)',\label{eq:ode_energy}
\end{align}
together with the boundary conditions
\[
\big(v(\pm\infty),u(\pm\infty),T(\pm\infty)\big)=\big(v_\pm,u_\pm,T_ 
\pm\big).
\]
Evidently, we can integrate each of the differential equations from $- 
\infty$ to $x$, and using the boundary conditions (in particular $v_- 
=1$ and $u_-=0$), we find, after some elementary manipulations, the  
profile equations:
\begin{align}
\mu v'&=v\Big[(v-1)+p(v,T)-p(v_-,T_-)\Big],\label{eq:profile1}\\
\kappa T'&=v\left[e(v,T)+\frac{(v-1)^2}{2}-e(v_-,T_-)\right]+v(v-1) 
\Big[p(v_-,T_-)-(v-1)\Big].\label{eq:profile2}
\end{align}
We note that in the case of an ideal gas, with $v_-=1$,
%and taking without loss of generality $\mu=1$,
these ODEs simplify somewhat, to
\begin{align}
v'&=\frac{1}{\mu}\left[v(v-1)+\Gamma (e-ve_-)\right],\label{eq:ideal_profile1}\\
%TODO: fixed! parentheses in wrong place! (DONE)
e'&=\frac{v}{\nu}\left[-\frac{(v-1)^2}{2}+(e-e_-)+(v-1)\Gamma e_-\right]
.\label{eq:ideal_profile2}
\end{align}
where $\nu:=\kappa/c_v$ and $\Gamma$ is as in \eqref{Gammaeq}.

\begin{remark}\label{goodprof}
Remarkably, the right-hand sides of the profile ODE are
polynomial in $(v,e)$, so smooth even for values on the
boundaries $\hat v=0$ or $\hat e=0$ of the physical parameter
range.  This is in sharp contrast to the isentropic case \cite{BHRZ,HLZ},
for which the ODE become singular as $v\to 0$, except in the special
case $\gamma=1$. 
\end{remark}

\subsection{Rankine-Hugoniot conditions}\label{RH}

Substituting $v_+,u_+,e_+$ into the rescaled profile equations 
\eqref{eq:ode_mass}--\eqref{eq:ode_energy}
and requiring that the right-hand side vanish
yields the Rankine-Hugoniot conditions
\begin{align}
-s[v] &=[u],\label{eq:rh_mass}\\
-s[u]&=-[p],\label{eq:rh_momentum}\\
-s\left[e+\frac{u^2}{2}\right]&=-[pu],\label{eq:rh_energy}
\end{align}
where $[f(U)]:=f(U_+)-f(U_-)$ denotes jump between $U_\pm$.

We specialize now to the ideal gas case. 
Under the scaling \eqref{scaling},
we have $s=-1$, $v_-=1$, $u_-=0$.
% $\mu=1$.
Fixing $\Gamma_{max}\ge \Gamma \ge \Gamma_{min}>0$  
%$\nu_{max}\ge \nu\ge \nu_{min}>0$, 
and letting $v_+$
vary in the range $1\ge v_+\ge v_*(\Gamma):= \Gamma/(\Gamma+2)$,
we use \eqref{eq:rh_mass}--\eqref{eq:rh_energy} to solve for
\[
u_+, e_+\;\text{and}\; e_-.
\]
Our assumptions reduce \eqref{eq:rh_mass}--\eqref{eq:rh_energy} to
\begin{align}
v_+-1 &=u_+,\label{eq:new_rh_mass}\\
u_+&=-(p_+-p_-)=-\Gamma\Big(\frac{e_+}{v_+}- e_-\Big),
\label{eq:new_rh_momentum}\\
(e_+-e_-)+\frac{u_+^2}{2}&=-p_+u_+=-\Gamma\frac{e_+u_+}{v_+}.\label 
{eq:new_rh_energy}
\end{align}
Equation \eqref{eq:new_rh_mass} immediately gives $u_+$.
Subtracting $\frac{u_+}{2}$ times \eqref{eq:new_rh_momentum} from
\eqref{eq:new_rh_energy}, and rearranging, we obtain
\begin{equation}\label{eq:R}
e_+= e_-
\frac{1+ \frac{\Gamma}{2}(1-v_+)}
{1-\frac{\Gamma}{2v_+}(1-v_+)}=
\frac{ e_- v_+}{(\Gamma +2)}
\frac{(\Gamma +2 -\Gamma v_+))}
{(v_+-v_*)},
\end{equation}
$v_*=\frac{\Gamma}{\Gamma+2}$,
from which we obtain the physicality condition
\begin{equation}\label{phys}
v_+> v_*:= \frac{\Gamma}{\Gamma + 2},
\end{equation}
corresponding to positivity of the denominator, with
$\frac{e_+}{e_-}\to +\infty$ as $v\to v_*$.
Finally, substituting into $1=s^2=-\frac{[p]}{[v]}$ and rearranging,
we obtain
\begin{equation}\label{e-}
e_-=\frac{(\Gamma +2)(v_+-v_*)}{2\Gamma(\Gamma +1)},
\end{equation}
and thus
\begin{equation}\label{e+}
e_+=\frac{v_+(\Gamma +2-\Gamma v_+)}{2\Gamma(\Gamma +1)},
\end{equation}
completing the description of the endstates.

We see from this analysis that the strong-shock limit corresponds,
for fixed $\Gamma$, to the limit $v_+\to v_*$, with all other
parameters functions of $v_+$.
In this limit,
\begin{equation}\label{asymptotics1}
v_-=1, \quad u_-=0,  \quad e_- \sim (v_+-v_*)\to 0,
\end{equation}
and
\begin{equation}\label{asymptotics2}
u_+\sim (v_+-1)\to \frac{-2}{\Gamma +2}, \quad
e_+ \to \frac{1- v_*^2}{2(\Gamma +1)}= \frac{2}{(\Gamma +2)^2}.
\end{equation}

At this point, taking without loss of generality $\mu=1$,
we have reduced to a three-parameter family of problems on compact
parameter range, parametrized by 
$\Gamma_{max}\ge \Gamma \ge \Gamma_{min}>0$,
$1\ge v_ +\ge v_*(\Gamma)\ge v_*(\Gamma_{min})>0$,
and $\nu_{max}\ge \nu \ge \nu_{min}$.

\begin{remark}\label{eblowup}
As $\Gamma \to 0$, we find from
\eqref{e+} that $e_+$ blows up as $(v_+-v_*)/\Gamma$, i.e.,
our rescaled coordinates remain bounded only if $v_+-v_*\le C\Gamma\to 0$, $C>0$ constant,
as well.
(This is reflected in the limiting profile equation for $\Gamma=0$,
which admits only profiles from $v_-=1$ to $v_+=0$;
see \eqref{eq:ideal_profile2}, which, for $\Gamma=0$, reduces
to $v'=v(v-1)$.)
Thus, our techniques apply for $\Gamma\to 0$ only
in the (simultaneous) large-amplitude limit.
\end{remark}

\subsection{Existence and decay of profiles}\label{hypcheck}
Specializing to the ideal gas case, we next study existence
and behavior of profiles.  Existence and exponential decay
of profiles has been established by Gilbarg \cite{Gi} 
for all {\it finite-amplitude shocks} $1\ge v_+>v_*$.
Thus, the question is whether these properties
extend to the strong-shock limit, the main issue being
to establish {\it uniform exponential decay} as $x\to \pm \infty$,
independent of shock strength.

Since the profile equations \eqref{eq:ideal_profile1}--\eqref{eq:ideal_profile2}
are smooth (polynomial) in $(v,e)$, the issue of uniform decay 
reduces essentially to {\it uniform hyperbolicity} of endstates $(v,e)_\pm$,
i.e., nonexistence of purely imaginary linearized
growth/decay rates at $\pm \infty$.
Linearizing \eqref{eq:ideal_profile1}--\eqref{eq:ideal_profile2}
about an equilibrium state, we obtain
\begin{equation}\label{linprof}
\begin{pmatrix}
v \\ e
\end{pmatrix}'=
M
\begin{pmatrix}
v \\ e
\end{pmatrix},
\qquad
M:=
\begin{pmatrix}
\mu^{-1} & 0\\
0 & v\nu^{-1}
\end{pmatrix}
\begin{pmatrix}
2v-1-\Gamma e_- & \Gamma \\
1-v+\Gamma e_- & 1\\
\end{pmatrix}.
\end{equation}
Since $M$ is $2\times 2$, its
eigenvalues are
$$
m= \frac{\trace M \pm \sqrt{\trace M^2 -4\det M}}{2},
$$
and so hyperbolicity is equivalent to 
$\det M \ne 0$  and $\det M < 0$  or $\trace M \ne 0$.

Computing, we have
\begin{equation}\label{hypcon}
\det M=
(v/\mu \nu)\Big(
(\Gamma +2)v-(\Gamma+1)(1+\Gamma e_-)
\Big),
\end{equation}
so that $\det M\gtrless 0$ is equivalent 
(for $\Gamma > 0$, hence $v\ge v_*>0$)
to
\begin{equation}\label{hypcon2}
(\Gamma +2)v-(\Gamma+1)(1+\Gamma e_-)\gtrless 0.
\end{equation}

At $v=v_-=1$, this reduces to $e_-\ne \frac{1}{\Gamma(\Gamma+1)}$,
or, using \eqref{e-} to 
$$
v_+ \frac{(\Gamma +2)(1-v_+)} {2}> 0,
$$
except in the characteristic case $v_+=1$, while 
$$
\trace M= \mu^{-1}(1-\Gamma e_- + (\nu/\mu)^{-1})\ge 
\mu^{-1}(1- \frac{\Gamma +2}{2(\Gamma+1)} + (\nu/\mu)^{-1})
\ge \nu>0.
$$
At $v=v_+$, \eqref{hypcon2} reduces, using \eqref{e-}, to
$$
\det M= (v_+/\mu \nu)\Big(
v_+ \frac{(\Gamma +2)(v_+-1)} {2} \Big)< 0,
$$
except in the characteristic case $v_+=1$.
Thus, for $v_+$ bounded from zero,
{\it hyperbolicity fails at $x=\pm\infty$
only in the characteristic case $v_-=v_+=1$}.

Next, let us recall the existence proof of \cite{Gi},
which proceeds by the observation that isoclines
$v'=0$ and $e'=0$ obtained by setting the right-hand sides
of \eqref{eq:ideal_profile1} and \eqref{eq:ideal_profile2} to zero
bound a convex lens-shaped region whose vertices are the
unique equilibria $U_\pm$, that is invariant under the forward flow
of \eqref{eq:ideal_profile1}--\eqref{eq:ideal_profile2}, and
into which enters the unstable manifold of $U_-$;
recall that $\det M<0$ at $-\infty$, hence there is a one-dimensional
unstable manifold.
It follows that the unstable manifold must approach attractor $U_+$
as $x\to +\infty$, determining the unique connecting orbit describing
the profile.

By the above-demonstrated hyperbolicity,
this argument extends also to the case $v_+=v_*$ ($e_-=0$),
yielding at once existence and uniform boundedness of profiles across
the whole parameter range (recall that
$(v,e)_\pm$ are uniformly bounded, 
by the Rankine--Hugoniot analysis of the previous section),
in particular for the limiting profile equations at $v_+=v_*$ of
\begin{align}
v'&=\frac{1}{\mu}\left[v(v-1)+\Gamma e\right],\label{eq:lim_profile1}\\
e'&=\frac{v}{\nu}\left[-\frac{(v-1)^2}{2}+e\right].
\label{eq:lim_profile2}
\end{align}

%that, (i) away from $v_+=0$ or the characteristic limit $v_+=1$, 
%we have uniform exponential convergence of $(\hat v, \hat u, \hat e)$ 
%to its endstates, and (ii) We have uniform convergence as $v_+\to v_*$ to
%a limiting profile $(\hat v^0, \hat u^0, \hat e^0)$ satisfying
%\eqref{eq:ideal_profile1}--\eqref{eq:ideal_profile2} with $v_+=v_*$, or

Collecting facts, we have the following key result.
%, as a consequence
%of which profiles vary smoothly over the range of investigation
%$v_+\in [1-\epsilon, v_*]$, $\epsilon >0$ corresponding to
%the strong-shock limit, with uniform exponential decay.

\begin{lemma}\label{profdecay}
For $\Gamma$ bounded and bounded away from the nonphysical limit $\Gamma=0$, 
$\mu, \nu$ bounded and bounded from zero, and $v_+$ bounded away from 
the characteristic limit $v_-=1$,
profiles $\hat U=(\hat v,\hat u,\hat e)^T$ of the rescaled
equations \eqref{eq:ideal_profile1}--\eqref{eq:ideal_profile2}
exist for all $1\ge v_+\ge v_*$,
decaying exponentially to their endstates $U_\pm$
as $x\to \pm \infty$, uniformly in $\Gamma, v_+, \mu, \nu$.
\end{lemma}

\begin{proof}
Existence, boundedness, and exponential decay of individual profiles 
follow from the discussion above.
Uniform bounds follow by smooth dependence on parameters together with
compactness of the parameter range.
\end{proof}

\section{Evans function formulation}\label{formulation}
From now on, we specialize to the ideal gas case,
%TODO: put back in the mu's, to help with temp. dependent case...
setting without loss of generality $\mu=1$. 

\subsection{Linearized integrated eigenvalue equations}\label{linint}
Defining integrated variables
\[
\tilde v:= \int_{-\infty}^x v\,dy,\;\; \tilde u:=\int_{-\infty}^x u 
\,dy, \;\;\tilde E:=\int_{-\infty}^x E\,dy,
\]
we note that the rescaled equations \eqref{eq:rescale_mass}-- 
\eqref{eq:rescale_energy} can be written in terms of the integrated  
variables in the form
\begin{equation}\label{integrated}
\begin{aligned}
\tilde v_t+\tilde v_x-\tilde u_x&=0,\\
\tilde u_t+\tilde u_x+\frac{\Gamma (\tilde E_x-\frac{\tilde u_x^2} 
{2})}{\tilde v_x}&=\frac{\tilde u_{xx}}{\tilde v_x},\\
\tilde E_t+\tilde E_x+
\frac{\Gamma \tilde u_x (\tilde E_x-\frac{\tilde u_x^2}{2})}{\tilde  
v_x}&=
\frac{ \tilde u_x\tilde u_{xx}}{\tilde v_x}+ \frac{\nu (\tilde E_x- 
\frac{\tilde u_x^2}{2})_x}{\tilde v_x}.
\end{aligned}
\end{equation}
%where we have replaced $v,u,E$ by
%$ \int_{-\infty}^x vdy, \int_{-\infty}^x udy, \int_{-\infty}^x Edy$.

The integrated viscous profile,
\[
\tilde{\hat{v}}:=\int_{-\infty}^x \hat v\,dy, \;\;\tilde{\hat{u}}:= 
\int_{-\infty}^x \hat u\,dy,\;\; \tilde{\hat{E}}:=\int_{-\infty}^x  
\hat E\,dy,
\]
is a stationary solution of \eqref{integrated}. Then we write
\[
\tilde v=\tilde{\hat{v}}+v,\;\; \tilde u=\tilde{\hat{u}}+u,\;\;\tilde  
E=\tilde{\hat{E}}+E,
\]
and we linearize \eqref{integrated} about the integrated profile. By  
an abuse of notation, we
denote the perturbation by $v$, $u$, and $E$. Note also that the  
integrated profile always appears under an $x$-derivative; this  
explains the appearance of ``hats'' in place of ``tilde-hats'' in the  
expression below. Finally, we use the relationship $\hat E=\hat e- 
\frac{\hat u^2}{2}$ to simplify some of the expressions.
We obtain the linearized integrated equations
\begin{equation}\label{linearized}
\begin{aligned}
v_t+v_x-u_x&=0,\\
u_t+u_x+\frac{\Gamma (E_x-\hat u u_x)}{\hat v}
-\frac{\Gamma \hat e}{\hat v^2} v_x
&=\frac{u_{xx}}{\hat v}-
\frac{\hat u_{x}}{\hat v^2}v_x,\\
E_t+E_x+
\frac{\Gamma \hat u (E_x-\hat u u_x)}{\hat v}
+\frac{\Gamma \hat e}{\hat v} u_x
-\frac{\Gamma \hat e\hat u}{\hat v^2} v_x
&=
\frac{ \hat uu_{xx}}{\hat v}
+\frac{ \hat u_x }{\hat v} u_x
-\frac{ \hat u\hat u_x }{\hat v^2} v_x
\\
& + \frac{\nu (E_x-\hat u u_x)_x}{\hat v}
- \frac{\nu \hat e_x}{\hat v^2}v_x.
\end{aligned}
\end{equation}
Defining $\epsilon:=E-\hat u u$, subtracting $\hat u$ times
the second equation from the third, and rearranging,
we obtain, finally, the linearized integrated eigenvalue equations:
\begin{equation}\label{evalue}
\begin{aligned}
&\lambda v+v'-u'=0,\\
&\lambda u+u'+\frac{\Gamma}{\hat v}\eps'
+\frac{\Gamma \hat u_x}{\hat v}u +\left[
-\frac{\Gamma \hat e}{\hat v^2}+
\frac{\hat u_{x}}{\hat v^2}\right]v'
=\frac{u''}{\hat v},\\
&\lambda \eps+\eps'+ \left[\hat u_x-\frac{\nu\hat u_{xx}}{\hat v} 
\right] u
+\left[\frac{\Gamma\hat e}{\hat v}-(\nu+1)\frac{\hat u_x}{\hat v} 
\right]u'+\left[\frac{\nu \hat e_x}{\hat v^2}\right]v'
=\frac{\nu}{\hat v}\eps''.
\end{aligned}
\end{equation}

\subsection{Expression as a first-order system}
Following \cite{BHRZ}, we may express \eqref{evalue} concisely
as a first-order system
\begin{equation}\label{firstorder}
W' = A(x,\lambda) W,
\end{equation}
\begin{equation} \label{evans_ode}
A(x,\lambda) =
\begin{pmatrix}
0 & 1 & 0 & 0 & 0 \\
\lambda\nu^{-1}\hat v & \nu^{-1}\hat v & \nu^{-1}\hat  v\hat u_x-\hat  
u_{xx} & \lambda g(\hat U) & g(\hat U)- h(\hat U) \\
0 & 0&  0 & \lambda & 1\\
0& 0& 0 & 0 & 1\\
0 & \Gamma & \lambda\hat v+\Gamma\hat u_x & \lambda\hat v & f(\hat U)- \lambda
\end{pmatrix},
\end{equation}
\begin{equation}
W = ( \eps ,\eps',u,v,v')^T, \quad \prime  
= \frac{\dif}{\dif x},
\end{equation}
where, using $\hat u_x=\hat v_x$ and \eqref{eq:ideal_profile1} with $ 
\mu=1$,
\begin{equation}
\begin{aligned}
g(\hat U)&:=\nu^{-1}(\Gamma\hat e-(\nu+1)\hat u_x) \\
	&=\nu^{-1}\Gamma\hat e--\frac{\nu+1}{\nu}\hat v_x \\
	&=\nu^{-1}\Gamma\hat e-\frac{\nu+1}{\nu}\Big(\hat v(\hat v-1)+\Gamma  
(\hat e-\hat ve_-)\Big)
\label{eq:g_eq}
\end{aligned}
\end{equation}
\begin{equation}
f(\hat U):=\frac{\hat u_x-\Gamma \hat e}{\hat v}+\hat v=
2\hat v-1-\Gamma e_-.
\label{eq:f_eq}
\end{equation}
\begin{equation}
h(\hat U):=-\frac{\hat e_x}{\hat v}= 
-\nu^{-1}\left(-\frac{(\hat v-1)^2}{2}+(\hat e-e_-)+(\hat v-1)\Gamma e_-\right).
\label{eq:h_eq}
\end{equation}

%TODO: irrelevant (and wrong), delete-K
%We note that 
%$$(g,f,h)(U_-)=(0,1,0), \quad
%(g,f,h)(U_+)=(\Gamma e_+/\nu,-\Gamma e_+/v_+,0)
%$$
%(hence $g_+>0$, $f_+<0$).

\begin{remark}
Remarkably, similarly as for the profile equations,
the entries of $A$ are polynomial in $(\hat v, \hat u, \hat e)$.
Thus, both profile and linearized eigenvalue equations are
perfectly well-behaved for any compact subset of $\Gamma>0$.
\end{remark}

\subsection{Consistent splitting}\label{splitting}
Denote by 
$A_\pm(\lambda):=\lim_{x\to \pm \infty} A(x,\lambda)$ 
the limiting coefficient matrices at $x=\pm \infty$.
(These limits exist by exponential convergence of
profiles $\hat U$, Lemma \ref{profdecay}.)
Denote by $S_\pm$ and $U_\pm$ the stable and unstable 
subspaces of $A_\pm$.

\begin{definition}
Following \cite{AGJ}, we say that \eqref{firstorder} exhibits
consistent splitting on a given $\lambda$-domain if
$A_\pm$ are hyperbolic, with $\dim S_+$ and $\dim U_-$
constant and summing to the dimension of the full space 
(in this case $5$).
\end{definition}

By analytic dependence of $A$ on $\lambda$ and 
standard matrix perturbation theory, $S^+$ and
$U_-$ are analytic on any 
domain for which consistent splitting holds.

\begin{lemma}\label{consistent}
For all $\Gamma >0$, $1 \ge v_+>v_*$, \eqref{firstorder}--\eqref{evans_ode}
exhibit consistent splitting on $\{\Re \lambda \ge 0\}\setminus \{0\}$,
with $\dim S_+=3$ and $\dim U_-=2$.
Moreover, subspaces $S_+$ and $U_-$, along with their associated
spectral projections, extend analytically in $\lambda$ and 
continuously in $\Gamma, \nu, v_-$, to 
$\{\Re \lambda \ge 0\}$ for $\Gamma >0$, $\nu>0$, and $1>v_+ \ge v_*$.
\end{lemma}

\begin{proof}
Consistent splitting and analytic extension to $\lambda=0$
follow by the general results of \cite{MZ.3}, except
at the strong-shock limit $v_+=v_*$, where
\begin{equation} \label{A-}
A^*_-(\lambda) =
\begin{pmatrix}
0 & 1 & 0 & 0 & 0 \\
\lambda\nu^{-1} & \nu^{-1} & 0 & 0 &
0 \\
0 & 0&  0 & \lambda & 1\\
0& 0& 0 & 0 & 1\\
0 & \Gamma & \lambda  & \lambda  & 1-\lambda
\end{pmatrix}
\end{equation}
and
\begin{equation} \label{A+}
A^*_+(\lambda) =
\begin{pmatrix}
0 & 1 & 0 & 0 & 0 \\
\lambda\nu^{-1}v_* & \nu^{-1}v_* & 0 & \lambda g(U_+) &
g(U_+) \\
0 & 0&  0 & \lambda & 1\\
0& 0& 0 & 0 & 1\\
0 & \Gamma & \lambda v_* & \lambda v_* & f(U_+)-\lambda
\end{pmatrix},
\end{equation}
where $g(U_+)=\frac{(1-v_*)}{\nu}(v_*- (\nu -1))$ and
$f(U_+)= \frac{\Gamma -2}{\Gamma +2}$.

The matrix $A^*_-$ is lower block-triangular, with diagonal blocks
$$
\begin{pmatrix}
0 & 1 \\
\lambda\nu^{-1} & \nu^{-1} \\
\end{pmatrix},
\qquad
\begin{pmatrix}
0 & \lambda & 1\\
0 & 0 & 1\\
\lambda  & \lambda  & 1-\lambda
\end{pmatrix}
$$
corresponding respectively to the scalar convection-diffusion equation
$$
e_t +e_x= \nu e_{xx}
$$
and the isentropic case treated in \cite{HLZ}, 
The first has eigenvalues $\mu= \frac{1\pm \sqrt{1+4\lambda}}{2}$,
so satisfies consistent splitting on 
$\{\Re \lambda\ge 0\}\setminus \{0\}$, with analytic continuation
(since eigenvalues remain separated) to $\Re \lambda \ge 0$.
The second, as observed in \cite{HLZ},
has eigenvalues $\mu=-\lambda, \frac{1\pm \sqrt{1+4\lambda}}{2}$,
hence likewise satisfies consistent splitting on $\{\Re \lambda \ge 0\}
\setminus \{0\}$ and (since the single unstable eigenvalue
remains separated from the two stable eigenvalues) continues analytically
to $\Re \lambda \ge 0$.
Indeed, the unstable manifold has dimension two for all $\Re \lambda \ge 0$,
hence is analytic on that domain.
This verifies the proposition at $x=-\infty$ by direct calculation.

At $x=+\infty$, the computation is more difficult.  Here, we
refer instead to the abstract results of \cite{MZ.3}, which
assert that hyperbolic--parabolic systems of the type treated
here, including the limiting case, at least for $v_*\ne 0$, 
exhibit consistent splitting on $\{\Re \lambda \ge 0\}\setminus \{0\}$,
with analytic extension to $\Re \lambda \ge 0$,
so long as the shock is noncharacteristic, i.e., the flux Jacobian
associated with the first-order part of the equations
have nonvanishing determinants at $x=\pm\infty$.
These may be computed in any coordinates, in particular
$(v,u,\epsilon)$.
Neglecting terms originating from diffusion, i.e., including only
first-order terms from the left-hand side, we obtain from
\eqref{evalue} the flux Jacobian
$$
\begin{pmatrix}
1 & -1 & 0\\
-\Gamma e/v^2 & 1 & \Gamma/v\\
0 & \Gamma e/v & 1
\end{pmatrix},
$$
which has determinant
$ \Delta=1 - \Gamma^2e/v^2- \Gamma e/v^2 $,
giving for $v_+=v_*$ ($e_-=0$) that
$ \Delta_{-\infty}=1> 0$ and, calculating at $v_+=v_*$ 
that
$$
\Gamma e_+/v_*^2=2,
\quad
\Delta_{+\infty}=-1-2\Gamma < 0.
$$
Thus, we may conclude by
the general results of \cite{MZ.3} that consistent splitting
holds at both $x=\pm \infty$ on $\{\Re \lambda \ge 0 \}
\setminus \{0\}$ for  $1\ge v_+\ge v_*$, with analytic
extension to $\Re \lambda \ge 0$.
\end{proof}

\begin{remark}\label{+rmk}
We note that the results of \cite{MZ.3} do not apply at $x=-\infty$,
$v_+=v_*$, where $e_-=0$ leaves the physical domain.
Specifically, at this value the genuine coupling condition of Kawashima
\cite{KSh}, hence the dissipativity condition of \cite{MZ.3} fails, and 
so we cannot conclude consistent splitting;
indeed, the eigenvalue $\mu \equiv \lambda$ 
(corresponding to the decoupled hyperbolic mode)
is pure imaginary for any pure imaginary $\lambda$.
\end{remark}

\subsection{Construction of the Evans function}\label{evans}
We now construct the Evans function associated with 
\eqref{firstorder}--\eqref{evans_ode}, following the
approach of \cite{MZ.3,PZ}.

\begin{lemma}\label{cbasis}
There exist bases $V^-=(V_1^-, V_2^-)(\lambda)$, 
$V^+=(V_3^+, V_4^+, V_5^+)(\lambda)$ of
$U_-(\lambda)$ and $S_+(\lambda)$, extending analytically in $\lambda$
and continuously in $\Gamma, \nu, v_-$ to
$\{\Re \lambda \ge 0\}$ for $\Gamma>0$, $\nu>0$, and $1>v_+\ge v_*$, determined
by Kato's ODE
\begin{equation}\label{kato}
V'=(PP'-P'P)V,
\end{equation}
where $P$ denotes the spectral projection onto $S_+$, $U_-$, respectively,
and $'$ denotes $d/d\lambda$.
\end{lemma}

\begin{proof}
This follows from Lemma \ref{consistent} by
a standard result of Kato \cite{Kato}, valid on any
simply connected set in $\lambda$ on which $P$ remains analytic.
\end{proof}

\begin{lemma}\label{basis}
There exist bases 
$$
W^-=(W_1^-, W_2^-)(\lambda), 
\quad
W^+=(W_3^+, W_4^+, W_5^+)(\lambda)
$$
of the unstable manifold at $-\infty$ and the stable manifold at $x=+\infty$
of \eqref{firstorder}--\eqref{evans_ode},
asymptotic to $e^{A_-x}V^-$ and $e^{A_+x}V^+$, respectively,
as $x\to \mp\infty$,
and extending analytically in $\lambda$ and continuously in 
$\Gamma, \nu, v_+$ to $\{\Re \lambda \ge 0\}$ 
for $\Gamma>0$, $\nu>0$, and $1>v_+\ge v_*$.
\end{lemma}

\begin{proof}
This follows, using the conjugation lemma
of \cite{MeZ}, by uniform exponential convergence of
$A$ to $A_\pm$ as $x\to \pm \infty$, Lemma \ref{profdecay}.
\end{proof}

\begin{definition}\label{evansdef}
The Evans function associated with \eqref{firstorder}--\eqref{evans_ode}
is defined as
\begin{equation}\label{evanseq}
D(\lambda):=\det(W^+, W^-)|_{x=0}.
\end{equation}
\end{definition}

\begin{proposition}\label{evanscont}
The Evans function $D(\cdot)$ is analytic in $\lambda$ and continuous
in $\Gamma, \nu, v_+$ on $\Re \lambda \ge 0$
and $\Gamma >0$, $\nu>0$, and $1>v_+\ge v_*$.  Moreover, on 
$\{\Re \lambda \ge 0\}\setminus\{0\}$, its zeros correspond in
location and multiplicity with eigenvalues of the integrated
linearized operator $\mathcal{L}$, or, equivalently with
solutions of \eqref{evalue} decaying at $x=\pm \infty$.
\end{proposition}

\begin{proof}
The first statement follows by Lemma \ref{basis}, the second
by a standard result of Gardner and Jones \cite{GJ.1, GJ.2},
valid on the region of consistent splitting.
\end{proof}

\begin{remark}\label{stronglimit}
Proposition \ref{evanscont} includes in passing the key information
that the Evans function converges in the strong-shock limit
$v_+\to v_*$ to the Evans function for the limiting system 
at $v_+=v_*$, uniformly on compact subsets of $\{\Re \lambda \ge 0\}$,
as illustrated numerically in Fig. \ref{fig1}.
\end{remark}

\begin{remark}\label{bestevans}
The specification in \eqref{kato} of initializing bases
at infinity is optimal in that it minimizes ``action'' in
a certain sense; see \cite{HSZ} for further discussion. 
In particular, {\rm for any constant-coefficient system,
the Evans function induced by Kato bases \eqref{kato}
is identically constant in $\lambda$}.
For, in this case, bases $W^+$ and $W^-$ are given at $x=0$
by the values $V^\pm$
prescribed in \eqref{kato}, and both evolve 
according to the same ODE, hence $W=(W^-,W^+)_{x=0}$
satisfies $W'=[P,P']W$ and $D(\lambda):=\det W\equiv {\rm constant}$
by Abel's Theorem and the fact that $\trace [P,P']=0$,
where $[P,P']:=PP'-P'P$.
\end{remark}

%TODO: keep the below new rmk?  (I think it's interesting
%and should somewhere be said, but...?-K)
\begin{remark}\label{unification}
More generally, $\det(V^-,V^+)\equiv {\rm constant }$ 
in the ``traveling pulse'' case $U_+=U_-$,
by the argument of Remark \ref{bestevans}, whence the
Evans function constructed here may be seen to agree with
the ``natural'' Evans function 
(independent of the choice of $V^\pm$)
\begin{equation}\label{PWnorm}
E(\lambda):=\frac{\det(W^+,W^-)|_{x=0}}{\det(V^+,V^-)}
=
\frac{\tilde W^+\cdot W^-|_{x=0}}{\tilde V^+\cdot V^-}
\end{equation}
defined in \cite{PW} for that case.
The latter may in turn be seen to agree with
a ($2$-modified) characteristic Fredholm determinant 
of the associated linearized operator $L$ about the wave \cite{GLM07},
formally equivalent to
$$
\mbox{\rm det}_2 (I+(L_0-\lambda)^{-1}(L-L_0))\sim
\frac{\det_2(L-\lambda)}
{ \det_2(L_0-\lambda)},
$$
where $L_0$ denotes the (constant-coefficient) 
linearized operator about the background state $U_\pm$.
Our construction by Kato's ODE thus {\rm gives a natural
extension to the traveling-front case of 
the canonical constructions of \cite{PW,GLM07}
in the traveling-pulse case}, neither of which generalizes
in obvious fashion to the traveling-front setting
(the difficulty in both cases coming from the fact that
$\det (V^+,V_-)$ may vanish).
\end{remark}

\section{High-frequency bounds}\label{high}

We now carry out the main technical work of the paper, establishing
the following uniform bounds on the size of unstable eigenvalues.

\begin{proposition}\label{hf}
Nonstable eigenvalues $\lambda$ of \eqref{evalue}, i.e., eigenvalues
with nonnegative real part, are confined for $\gamma>1$, 
$v_*<v_+\le 1$ to a finite region $|\lambda|\le \Lambda$, for any
\begin{equation}\label{Lambda}
\Lambda  \ge 2\max\{1, \nu\} \max_{x} 
\Big( 
\frac{ |\cF^*_{--}|+|\cF^*_{++}| + 2\sqrt{|\cF^*_{-+}||\cF^*_{+-}|}}
{\hat v^{1/2}} 
(x, \Lambda)
\Big)^2,
\end{equation}
where
\begin{equation}\label{F*}
|\cF^*_{kl}|(x, \Lambda):=
\sum_{i=0}^4 \frac{ |\cF_{-i/2,kl}|}{|\Lambda|^{i/2}}(x),
\end{equation}
$k,l=+,-$,
$\cF_{j,kl}$ are as defined in \eqref{cFexp} below,
and $|\cdot|$ is the matrix operator norm with respect to any specified
norm on $\mathbb{C}^5$.
\end{proposition}

Before establishing Proposition \ref{hf}, we give a general discussion
indicating the ideas behind the proof.
For $v_+>v_*$, such high-frequency bounds have already
been established by asymptotic ODE estimates in \cite{MZ.3}.
For $v_+=v_*$, the problem leaves the class studied
in \cite{MZ.3} (specifically, the dissipativity condition 
is neutrally violated as discussed in Remark \ref{+rmk}), 
hence requires further discussion.

However, a brief examination reveals that the argument of 
\cite{MZ.3} applies in this case almost unchanged.
For, recall that the method of \cite{MZ.3} to obtain high-frequency
bounds was to decompose the flow of the first-order eigenvalue
equation for high frequencies into parabolic growth and decay
modes of equal dimensions $r=\dim (u,e)$ with growth
rates $\Re \mu \sim \pm |\lambda|^{1/2}$, 
and hyperbolic modes of dimension $n-r$, in the present case 
$\dim v=1$, with growth rate $\sim \pm (\Re \lambda  + 1)$ 
up to an exponentially decaying error term $\sim e^{-\theta|x|}$,
the delicate point being to separate decaying from growing hyperbolic
modes.

In the present, degenerate case, the hyperbolic rates
are only $\sim \Re \lambda$ plus decaying error term, 
and so the final, delicate part of the argument in \cite{MZ.3} does not apply.
However, since there is only a single hyperbolic mode, this part
of the argument is not needed.
Specifically, the $|\lambda|^{1/2}/C$ spectral gap between parabolic
and hyperbolic modes allow us for high frequencies to decompose the
flow of the eigenvalue equation into the direct sum of
growing parabolic modes blowing
up exponentially at $x=+\infty$, decaying parabolic modes 
blowing up exponentially at $-\infty$,
and a single hyperbolic mode that blows up exponentially 
at $-\infty$ for $v_+>v_*$ and, though it does not blow up
exponentially for $v_+=v_*$, is in any case always transverse
to the unstable manifold at $x=-\infty$.

To put things another way, the unstable manifold at $x=-\infty$
consists for $|\lambda|$ sufficiently large
precisely of growing parabolic modes, which blow up at $x=+\infty$.
Thus, there exist no zeros of the Evans function, since these
correspond to solutions belonging to both the unstable manifold
at $-\infty$ and the stable (hence decaying) manifold at $+\infty$.
This shows the existence of uniform high-frequency bounds- it
remains to establish quantitative bounds by keeping track of constants
throughout the argument.

\begin{remark}\label{exthf}
A review of the above shows that the same argument applies
whenever hyperbolic modes are uniformly decaying or growing,
i.e., in the situation identified in \cite{MZ.4,Z.3,GMWZ.1} 
that all hyperbolic characteristic speeds have a common sign.
Likewise, the multidimensional case may be treated by essentially
the same argument, following the generalization given in \cite{GMWZ.1}.
\end{remark}

%TODO: rmk somewhere that det one transformation between
%these and previous coord's?  NO- NOT NECESSARY... -KZ
\begin{proof}[Proof of Proposition \ref{hf}]
We carry out the argument in two steps.

{\bf 1. Preparation.}
Recasting \eqref{firstorder}, 
\eqref{evans_ode} in the standard coordinates of \cite{MZ.3}, as
\begin{equation}\label{stdfirstorder}
Z' = B(x,\lambda) Z,
\quad
Z = (v,  u, \eps,  u', \eps')^T,
\end{equation}
\begin{equation} \label{stdevans_ode}
B(x,\lambda) =
\begin{pmatrix}
-\lambda & 0 & 0 & 1 & 0 \\
0 & 0&  0 &  1& 0\\
0 & 0&  0 &  0& 1\\
\lambda (f(\hat U)-\hat v) & \lambda\hat v+\Gamma\hat u_x & 0 & f(\hat U) & \Gamma\\
\lambda h(\hat U) &  \nu^{-1}\hat  v\hat u_x-\hat  u_{xx} 
& \lambda\nu^{-1}\hat v & g(\hat U) - h(\hat U) & \nu^{-1}\hat  v\\
\end{pmatrix},
\end{equation}
and we decompose $B$ as $B=\lambda B_1 + B_0$ with
\begin{equation} \label{B1}
B_1(x) =
\begin{pmatrix}
-1 & 0 & 0 & 0 & 0 \\
0& 0& 0 & 0 & 0\\
0& 0& 0 & 0 & 0\\
f(\hat U)-\hat v& \hat v & 0 & 0 & 0\\
h(\hat U) & 0 &  \nu^{-1}\hat v & 0 & 0 & \\
\end{pmatrix},
\end{equation}
and
\begin{equation} \label{B0}
B_0(x) =
\begin{pmatrix}
0 & 0 & 0 & 1 & 0 \\
0 & 0&  0 &  1& 0\\
0 & 0&  0 &  0& 1\\
0 & \Gamma\hat u_x & 0 & f(\hat U) & \Gamma\\
0 &  \nu^{-1}\hat  v\hat u_x-\hat  u_{xx} 
& 0 & g(\hat U) - h(\hat U) & \nu^{-1}\hat  v\\
\end{pmatrix}.
\end{equation}

Noting that $B_1$ is lower block-triangular, with ($1\times 1$)
upper diagonal block $-1$ strictly negative, and ($4\times 4$)
lower diagonal block 
$$
\alpha:=
\begin{pmatrix}
0 & 0 & 0 & 0  \\
0 & 0 & 0 & 0  \\
\hat v & 0& 0 & 0 \\
0 & \nu^{-1}\hat v & 0 & 0 \\
\end{pmatrix},
$$
having all zero eigenvalues, we block-diagonalize
by the lower block-diagonal transformation $Z:=T\mathcal{X}$,
$$
T:=
\begin{pmatrix}
1 & 0 \\
\theta & I_4\\
\end{pmatrix},
\quad
T^{-1}:=
\begin{pmatrix}
1 & 0 \\
-\theta & I_4\\
\end{pmatrix},
$$
$$
\theta:= 
-(\alpha+I_4)^{-1}
\begin{pmatrix} 0\\ 0\\ f(\hat U)-\hat v\\ h(\hat U) \\ \end{pmatrix},
$$
where, since $\alpha^2=0$,
$(I+\alpha)^{-1}=I -\alpha$, hence
$$
\theta=
-\begin{pmatrix} 0\\ 0\\ f(\hat U)-\hat v\\ h(\hat U) \\ \end{pmatrix}
+\alpha
\begin{pmatrix} 0\\ 0\\ f(\hat U)-\hat v\\ h(\hat U) \\ \end{pmatrix}=
\begin{pmatrix} 0\\ 0\\ -f(\hat U)+\hat v\\ -h(\hat U) \\ \end{pmatrix},
$$
and
$$
T^{-1}T'=
\begin{pmatrix}
0 & 0 \\
\theta'& 0\\
\end{pmatrix}
=
\begin{pmatrix}
0 & 0 & 0 & 0 & 0 \\
0 & 0 & 0 & 0 & 0 \\
0 & 0 & 0 & 0 & 0 \\
-\hat v_x & 0 & 0 & 0 & 0 \\
j(\hat U) & 0 & 0 & 0 & 0 \\
\end{pmatrix},
$$
$j(\hat U):= \nu^{-1}\Big((\Gamma e_--(\hat v-1))\hat v_x +\hat e_x \Big)$,
to obtain $\mathcal{X}'=C\mathcal{X}$, 
$$
C= T^{-1}BT - T^{-1}T'= \lambda C_1 +C_0,
$$
where
\begin{equation} \label{C1}
C_1(x,\lambda) =
\begin{pmatrix}
-1 & 0 & 0 & 0 & 0 \\
0& 0& 0 & 0 & 0\\
0& 0& 0 & 0 & 0\\
0& \hat v & 0 & 0 & 0\\
0 & 0 &  \nu^{-1}\hat v & 0 & 0 & \\
\end{pmatrix},
\end{equation}
is in a variant of block Jordan form and
\begin{equation} \label{C0}
C_0(x,\lambda) =
\begin{pmatrix}
\hat v -f(\hat U) & 0 & 0 & 1 & 0 \\
\hat v -f(\hat U) & 0&  0 &  1& 0\\
-h(\hat U) & 0&  0 &  0& 1\\
k(\hat U) & \Gamma\hat u_x & 0 & 2f(\hat U)-\hat v & \Gamma\\
l(\hat U)
&  \nu^{-1}\hat  v\hat u_x-\hat  u_{xx} 
& 0 & g(\hat U) & \nu^{-1}\hat  v\\
\end{pmatrix},
\end{equation}
where 
\begin{equation}\label{kl}
\begin{aligned}
k(\hat U)&:= (2f(\hat U)-\hat v)(\hat v-f(\hat U))-\Gamma h(\hat U)+\hat v_x,\\
l(\hat U)&:= g(\hat U)(\hat v-f(\hat U))-\nu^{-1}\hat vh(\hat U)-j(\hat U).
\end{aligned}
\end{equation}

Making the further transformation $\mathcal{X}= Q\mathcal{Y}$, 
$$
Q:=
\begin{pmatrix}
1 & \beta \\
0 & I_4\\
\end{pmatrix},
\quad
Q^{-1}:=
\begin{pmatrix}
1 & -\beta \\
0 & I_4\\
\end{pmatrix},
$$
$$
\begin{aligned}
\beta&:= 
\lambda^{-1}(0, 0, 1, 0) (I_4+\alpha)^{-1}\\
&= \lambda^{-1}(0, 0, 1, 0) (I_4-\alpha)\\
&=\lambda^{-1}(-\hat v, 0, 1, 0),\\ 
\end{aligned}
$$
$$
Q^{-1}Q'=
\begin{pmatrix}
0 & \beta' \\
0 & 0\\
\end{pmatrix}
=
\lambda^{-1}\begin{pmatrix}
0 & -\hat v_x & 0 & 0 & 0 \\
0 & 0 & 0 & 0 & 0 \\
0 & 0 & 0 & 0 & 0 \\
0 & 0 & 0 & 0 & 0 \\
0 & 0 & 0 & 0 & 0 \\
\end{pmatrix},
$$
we obtain $\mathcal{Y}'=D\mathcal{Y}$, 
$$
\begin{aligned}
D&= Q^{-1}CQ - Q^{-1}Q'\\
&= 
\lambda D_1 + D_0 + \lambda^{-1}D_{-1}+ \lambda^{-2}D_{-2},
\end{aligned}
$$
where $D_1=C_1$,
\begin{equation} \label{D0}
\begin{aligned}
D_0(x,\lambda) =
\begin{pmatrix}
\hat v -f(\hat U) & 0 & 0 & 0 & 0 \\
\hat v -f(\hat U) & 0&  0 &  1& 0\\
-h(\hat U) & 0&  0 &  0& 1\\
k(\hat U) & \Gamma\hat u_x & 0 & 2f(\hat U)- \hat v & \Gamma\\
l(\hat U)
&  \nu^{-1}\hat  v\hat u_x-\hat  u_{xx} 
& 0 & g(\hat U) & \nu^{-1}\hat  v\\
\end{pmatrix},
\end{aligned}
\end{equation}
\begin{equation} \label{D-1}
\begin{aligned}
D_{-1}(x,\lambda) =
\begin{pmatrix}
m(\hat U) & -\hat v^2+f(\hat U)\hat v-\Gamma \hat u_x+\hat v_x & 0 & 3(\hat v-f(\hat U)) & -\Gamma \\
0 & -\hat v(\hat v -f(\hat U)&  0 &  \hat v- f(\hat U)& 0\\
0 & -\hat v h(\hat U)&  0 &  -h(\hat U)& 0\\
0 & -\hat v k(\hat U) & 0 & k(\hat U) &0\\
0 &  -\hat v  l(\hat U) & 0 &  l(\hat U) & 0\\
\end{pmatrix},
\end{aligned}
\end{equation}
and
\begin{equation} \label{D-2}
\begin{aligned}
D_{-2}(x,\lambda) =
\begin{pmatrix}
0 & -\hat v m(\hat U) & 0 & m(\hat U) & 0 \\
0 & 0 & 0 & 0 & 0\\
0 & 0 & 0 & 0 & 0\\
0 & 0 & 0 & 0 & 0\\
0 & 0 & 0 & 0 & 0\\
\end{pmatrix},
\end{aligned}
\end{equation}
where
$$
m(\hat U):= \hat v(\hat v-f(\hat U))-k(\hat U).
$$

Making the ``balancing'' transformation $\mathcal{Y}=V\mathcal{Z}$,
$V=\begin{pmatrix} I_3 & 0 \\
0 &  \lambda^{1/2} I_2  \end{pmatrix}$,
we then obtain $\mathcal{Z}'=E\mathcal{Z}$, 
$ E= V^{-1}DV $,
where
\begin{equation}\label{E}
E=\begin{pmatrix} 
-\lambda & 0 & 0 & 0 & 0\\
0 & 0 &     0  & \lambda^{1/2} & 0\\
0 & 0 & 0 & 0 & \lambda^{1/2}\\
0 & \lambda^{1/2} \hat v & 0 & 0 & 0 \\
0 & 0 & \lambda^{1/2} \nu^{-1}\hat v & 0 & 0 \\
\end{pmatrix}
+\Theta,
\end{equation}
$$
\Theta=\Theta_0 + \lambda^{-1/2} \Theta_{-1/2}
+\lambda^{-1} \Theta_{-1}
+\lambda^{-3/2} \Theta_{-3/2}
+\lambda^{-2} \Theta_{-2},
$$
where 
\begin{equation}\label{Theta0}
\Theta_0(x) =
\begin{pmatrix}
\hat v -f(\hat U) & 0&  0 &  0& 0\\
\hat v -f(\hat U) & 0&  0 &  0& 0\\
-h(\hat U) & 0&  0 &  0& 0\\
0 & 0 & 0 &2f(\hat U)- \hat v & \Gamma\\
0 &  0 & 0 & g(\hat U) & \nu^{-1}\hat  v\\
\end{pmatrix},
\end{equation}
\begin{equation}\label{Theta-1/2}
\Theta_{-1/2}(x) =
\begin{pmatrix}
0 & 0 &  0 & 3(\hat v-f(\hat U))  & -\Gamma\\
0 & 0&  0 &  \hat v- f(\hat U)& 0\\
0 & 0 &  0 &  -h(\hat U)& 0\\
k(\hat U) & \Gamma\hat u_x & 0 & 0 & 0\\
l(\hat U) &  \nu^{-1}\hat  v\hat u_x-\hat  u_{xx} 
& 0 & 0 & 0\\
\end{pmatrix},
\end{equation}
\begin{equation}\label{Theta-1}
\Theta_{-1}(x) =
\begin{pmatrix}
m(\hat U) & -\hat v(\hat v-f(\hat U))+ \hat v_x-\Gamma \hat u _x& 0 & 0 & 0 \\
0 & -\hat v(\hat v -f(\hat U)&  0 &  0 & 0\\
0 & \hat v h(\hat U)&  0 &  0 & 0\\
0 & 0 & 0 & k(\hat U) &0\\
0 &  0  & 0 &  l(\hat U) & 0\\
\end{pmatrix},
\end{equation}
\begin{equation}\label{Theta-3/2}
\Theta_{-3/2}(x) =
\begin{pmatrix}
0 & 0 &  0 & m(\hat U)  & 0\\
0 & 0 &  0 & 0  & 0\\
0 & 0 &  0 & 0  & 0\\
0 & -\hat v k(\hat U) & 0 & 0 & 0\\
0 &  -\hat v  l(\hat U) & 0 &  0 & 0\\
\end{pmatrix},
\end{equation}
\begin{equation}\label{Theta-2}
\Theta_{-2}(x) =
\begin{pmatrix}
0 & -\hat v m(\hat U) & 0 & 0 & 0 \\
0 & 0 &  0 & 0 & 0\\
0 & 0 &  0 & 0 & 0\\
0 & 0 &  0 & 0 & 0\\
0 & 0 &  0 & 0 & 0\\
\end{pmatrix}.
\end{equation}

Finally, setting $\mathcal{Z}= \tilde S X$, with
$$
\tilde S:=
\begin{pmatrix}
1&0\\
0 & \tilde s\\
\end{pmatrix},
\quad
\tilde s:=
\begin{pmatrix}
I & I\\
-A & A
\end{pmatrix},
\quad
\tilde s^{-1}:=
\frac{1}{2}
\begin{pmatrix}
I & -A^{-1}\\
I & A^{-1}\\
\end{pmatrix},
$$
$$
A:=
\begin{pmatrix}
\sqrt{\hat v} & 0\\
0 & \sqrt{\nu^{-1}\hat v} \\
\end{pmatrix},
\quad
A^{-1}:=
\begin{pmatrix}
1/\sqrt{\hat v} & 0\\
0 & 1/\sqrt{\nu^{-1}\hat v} \\
\end{pmatrix},
$$
$$
\tilde S^{-1}\tilde S_x=
\begin{pmatrix}
0 & 0\\
0 & \tilde s^{-1}\tilde s_x \\
\end{pmatrix}
\quad
\tilde s^{-1}\tilde s_x=
(\hat v_x/4\hat v)
\begin{pmatrix}
I & -I\\
-I & I
\end{pmatrix},
$$
we obtain 
\begin{equation}\label{Xeq}
X'= (F+\mathcal{F}) X, 
\end{equation}
where 
\begin{equation}\label{F}
F= \tilde S^{-1}E\tilde S= \begin{pmatrix} M_- & 0 \\0 &  M_+\end{pmatrix}
\end{equation}
with
\begin{multline}\label{Ms}
M_+:= \begin{pmatrix} 
\lambda^{1/2}\hat v^{1/2} & 0\\
0 & \lambda^{1/2}(\hat v/\nu)^{1/2} \\
\end{pmatrix},
\\
M_-:=
\begin{pmatrix}
 -\lambda & 0 & 0 \\
0 &  -\lambda^{1/2} \hat v^{1/2}  & 0 \\
0 & 0 & -\lambda^{1/2}(\hat v/\nu)^{1/2} \\
\end{pmatrix},
\end{multline}
and
\begin{equation}\label{cFexp}
\begin{aligned}
\mathcal{F}&= \tilde S^{-1}\Theta \tilde S -\tilde S^{-1}\tilde S_x\\
&=
\cF_0 + \lambda^{-1/2} \cF_{-1/2}
+\lambda^{-1} \cF_{-1}
+\lambda^{-3/2} \cF_{-3/2}
+\lambda^{-2} \cF_{-2},
\end{aligned}
\end{equation}
where 
\begin{multline}\label{cF0}
\cF_0(x) = \\
\begin{pmatrix}
\hat v -f(\hat U) & 0&  0 &  0& 0\\
\frac{\hat v-f(\hat U)}{2} & 
\frac{2f(\hat U)-\hat v}{2}-\frac{\hat v_x}{4\hat v} & 
\frac{\sqrt{\nu^{-1}}\Gamma}{2}& 
-\frac{2f(\hat U)-\hat v}{2}+\frac{\hat v_x}{4\hat v}&  
\frac{-\sqrt{1/\nu}\Gamma}{2}\\
\frac{-h(\hat U)}{2} & 
\frac{g(\hat U)}{2\sqrt{1/\nu}} & 
\frac{\nu^{-1}\hat  v}{2} - \frac{\hat v_x}{4} &
-\frac{g(\hat U)}{2\sqrt{1/\nu}}& 
-\frac{\nu^{-1}\hat  v}{2} + \frac{\hat v_x}{4}	\\
\frac{\hat v -f(\hat U)}{2} & 
-\frac{2f(\hat U)-\hat v}{2}+\frac{\hat v_x}{4\hat v}  &
 -\frac{\sqrt{1/\nu}\Gamma}{2}& 
\frac{2f(\hat U)-\hat v}{2}-\frac{\hat v_x}{4\hat v}&  
\frac{\sqrt{1/\nu}\Gamma}{2}\\
\frac{-h(\hat U)}{2} &  
-\frac{g(\hat U)}{2\sqrt{1/\nu}}& 
- \frac{\nu^{-1}\hat  v}{2} + \frac{\hat v_x}{4} &
\frac{g(\hat U)}{2\sqrt{1/\nu}}& 
\frac{\nu^{-1}\hat  v}{2} - \frac{\hat v_x}{4}	\\
\end{pmatrix},\\
\end{multline}
%f minus one half equation
\begin{multline}\label{cF-1/2}
\cF_{-1/2}(x) = \\
\begin{pmatrix}
0 &
-3\sqrt{\hat v}(\hat v-f(\hat U)) &
\Gamma\sqrt{\frac{\hat v}{\nu}} &
3\sqrt{\hat v}(\hat v-f(\hat U)) & 
-\Gamma\sqrt{\frac{\hat v}{\nu}}\\
\frac{-k(\hat U)}{2\sqrt{\hat v}} & 
\frac{-\Gamma \hat u_x}{2\sqrt{\hat v}}-\frac{\sqrt{\hat v}}{2}(\hat v-f(\hat U))&  
0 &
\frac{-\Gamma \hat u_x}{2\sqrt{\hat v}}+\frac{\sqrt{\hat v}}{2}(\hat v-f(\hat U))& 
0\\
\frac{-l(\hat U)}{2\sqrt{\hat v/\nu}} & 
-\frac{n(\hat U)}{2\sqrt{\nu^{-1}\hat v}}+\frac{h(\hat U)\sqrt{\hat v}}{2}&  
0 &
-\frac{n(\hat U)}{2\sqrt{\nu^{-1}\hat v}}-\frac{h(\hat U)\sqrt{\hat v}}{2}&
0\\
\frac{k(\hat U)}{2\sqrt{\hat v}} & 
\frac{\Gamma \hat u_x}{2\sqrt{\hat v}}-\frac{\sqrt{\hat v}}{2}(\hat v-f(\hat U))&  
0 &
\frac{\Gamma \hat u_x}{2\sqrt{\hat v}}+\frac{\sqrt{\hat v}}{2}(\hat v-f(\hat U))& 
0\\
\frac{l(\hat U)}{2\sqrt{\hat v/\nu}} & 
\frac{n(\hat U)}{2\sqrt{\hat v/\nu}}+\frac{h(\hat U)\sqrt{\hat v}}{2}&  
0 &
\frac{n(\hat U)}{2\sqrt{\hat v/\nu}}-\frac{h(\hat U)\sqrt{\hat v}}{2}&
0
\end{pmatrix},
\end{multline}
$$
n(\hat U):= \nu^{-1}\hat v\hat u_x-\hat u_{xx},
$$
\begin{equation}\label{cF-1}
\cF_{-1}(x) =
\begin{pmatrix}
m(\hat U) &
q(\hat U) &
0  & 
q(\hat U)&
 0 \\
0 &
\frac{-\hat v(\hat v-f(\hat U))+k(\hat U)}{2} &
0 &  
\frac{-\hat v(\hat v-f(\hat U))-k(\hat U)}{2}&
0\\
0 &
\frac{\hat v h(\hat U)}{2}+\frac{l(\hat U)}{2\sqrt{\nu^{-1}}}&
0 &  
\frac{\hat v h(\hat U)}{2}-\frac{l(\hat U)}{2\sqrt{\nu^{-1}}} &
0 \\
0 &
\frac{-\hat v(\hat v-f(\hat U))-k(\hat U)}{2} &
0 &  
\frac{-\hat v(\hat v-f(\hat U))+k(\hat U)}{2}&
0\\
0 &
\frac{\hat v h(\hat U)}{2}-\frac{l(\hat U)}{2\sqrt{\nu^{-1}}}&
0 &  
\frac{\hat v h(\hat U)}{2}+\frac{l(\hat U)}{2\sqrt{\nu^{-1}}} &
0 \\
\end{pmatrix},
\end{equation}
\begin{equation}
q(\hat U):=-\hat v(\hat v-f(\hat U))-\Gamma \hat u_x+\hat v_x,
\end{equation}
\begin{equation}\label{cF-3/2}
\cF_{-3/2}(x) =
\begin{pmatrix}
0 & -\hat v^{1/2} m(\hat U)  & 0 & \hat v^{1/2} m(\hat U)  & 0\\
0 & \frac{\hat v^{1/2} k(\hat U)}{2} &  0 & \frac{\hat v^{1/2} k(\hat U)}{2} & 0\\
0 & \frac{\hat v^{1/2}l(\hat U)}{2\sqrt{\nu^{-1}}} &  0 &\frac{\hat v^{1/2}l(\hat U)}{2\sqrt{\nu^{-1}}}   & 0\\
0 & -\frac{\hat v^{1/2} k(\hat U)}{2} &  0 & \frac{-\hat v^{1/2} k(\hat U)}{2} & 0\\
0 & -\frac{\hat v^{1/2}l(\hat U)}{2\sqrt{\nu^{-1}}} &  0 & -\frac{\hat v^{1/2}l(\hat U)}{2\sqrt{\nu^{-1}}} & 0\\
\end{pmatrix},
\end{equation}
\begin{equation}\label{cF-2}
\cF_{-2}(x) =
\begin{pmatrix}
0 & -\hat v m(\hat U) & 0 &  -\hat v m(\hat U) & 0 \\
0 & 0 &  0 & 0 & 0\\
0 & 0 &  0 & 0 & 0\\
0 & 0 &  0 & 0 & 0\\
0 & 0 &  0 & 0 & 0\\
\end{pmatrix}.
\end{equation}

{\bf 2. Tracking.}
Denoting $X_-=(X_1,X_2, X_3)^T$, $X_+=(X_4,X_5)^T$, and 
$$
\cF=\begin{pmatrix} \cF_{--} & \cF_{-+}\\
\cF_{+-} & \cF_{++}\\
\end{pmatrix},
$$
we obtain from \eqref{Xeq}--\eqref{F}
\begin{equation}
\begin{aligned}
|X_-|'&\le 
 |\cF_{--}||X_-|+ |\cF_{-+}||X_+|,\\
|X_+|'&\ge  
\min \{1 , \nu^{-1/2} \}\hat v^{1/2}\Re \lambda^{1/2} |X_+|
-|\cF_{+-}||X_-| - |\cF_{++}||X_+|,\\
\end{aligned}
\end{equation}
from which, defining $\zeta:= |X_-|/|X_+|$, 
we obtain by a straightforward computation the Ricatti equation
\begin{equation}\label{Ricatti}
\zeta'\le \Big(-\min \{1 , \nu^{-1/2} \}\hat v^{1/2}\Re \lambda^{1/2} 
+|\cF_{--}|+|\cF_{++}|\Big)\zeta
+|\cF_{-+}|+ |\cF_{+-}|\zeta^2.
\end{equation}

Denote by
\begin{equation}\label{zeta+-}
\begin{aligned}
\zeta_\pm&:=
\frac{ \min \{1 , \nu^{-1/2} \}\hat v^{1/2}\Re \lambda^{1/2}  
-|\cF_{--}|-|\cF_{++}|}{2|\cF_{+-}|}\\
&\quad \pm
\sqrt{
\Big(\frac{ \min \{1 , \nu^{-1/2} \}\hat v^{1/2}
\Re \lambda^{1/2}  -|\cF_{--}|-|\cF_{++}|}{2|\cF_{+-}|}\Big)^2
-\frac{|\cF_{-+}|}{|\cF_{+-}|} }
\end{aligned}
\end{equation}
the roots of
\begin{equation}\label{quad}
\Big(-\min \{1 , \nu^{-1/2} \}\hat v^{1/2}\Re \lambda^{1/2} 
+|\cF_{--}|+|\cF_{++}|\Big)\zeta
+|\cF_{-+}|+ |\cF_{+-}|\zeta^2=0.
\end{equation}
Assuming for all $x$ the condition
\begin{equation}\label{goodcondition}
\max\{1, \nu^{1/2}\}
\frac{ |\cF_{--}|+|\cF_{++}| + 2\sqrt{|\cF_{-+}||\cF_{+-}|}}
{\hat v^{1/2}} 
< \Re \lambda^{1/2},
\end{equation}
$\zeta_\pm$ are positive real and distinct, whence, consulting
\eqref{Ricatti}, we see that $\zeta'<0$ on the interval
$\zeta_-<\zeta<\zeta_+$.

It follows that $\Omega_-:=\{\zeta\le \zeta_-\}$ is an invariant region
under the forward flow of \eqref{Xeq}; moreover, this region
is exponentially attracting for $\zeta < \zeta_+$.
A symmetric argument yields that $\Omega_+:=\{\zeta \ge \zeta_+\}$ is
invariant under the backward flow of \eqref{Xeq}, and exponentially
attracting for $\zeta >\zeta_-$.
Specializing these observations to the constant-coefficient limiting
systems at $x=-\infty$ and $x=+\infty$, we find that the invariant 
subspaces of the limiting coefficient matrices from which the
Evans function is constructed must lie in 
$\Omega_-$ and $\Omega_+$, respectively.
%NOTE, Details of this last point: obviously, genuine eigenvectors
%preserved, so can't be out of this cone.  More generally, Jordan
%block case, generalized eigenvectors NOT preserved, but move toward
%genuine ones only algebraically in time, contradicting exponential
%approach to cone. (this appears in multi-d appendix, so need not repeat..)
By  forward (respectively. backward) invariance of $\Omega_-$ (respectively. $\Omega_+$),
under the full, variable-coefficient flow, we thus find that the manifold
$\Span W^-$ of solutions initiated at $x=-\infty$ in the construction
of the Evans function
%TODO: make sure description of Evans construction agrees in notation
%with this passage here... -KZ
lies in $\Omega_-$ for all $x$, while the manifold $\Span W^+$ of 
solutions initiated at $x=+\infty$ lies in $\Omega_+$ for all $x$.

Since $\Omega_-$ and $\Omega_+$ are distinct,
we may conclude that {under condition \eqref{goodcondition},
$\Span W^+$ and $\Span W^-$ are transverse and
the Evans function does not vanish}.
But \eqref{Lambda} implies \eqref{goodcondition},
by $\Re \lambda^{1/2}\ge \frac{|\lambda|^{1/2}}{\sqrt{2}}$ together with
\eqref{cFexp}.
\end{proof}

\subsection{Universality and convergence in the high-frequency limit}\label{hfconv}
The bounds obtaining from \eqref{Lambda} may in practice be 
rather conservative, as illustrated by the following example.

\begin{example}\label{trackeg}
For the constant-coefficient 
scalar operator $L:= \partial_x^2 - a\partial_x $,
write $(L-\lambda)w=0$ as a first-order system $W'=(A+\Theta)W$,
$W=(w,w'/\lambda^{1/2})^T$, where
$
A=\lambda^{1/2}\begin{pmatrix}0&1\\ 1 & 0\end{pmatrix}$,
$\Theta= \begin{pmatrix}0&0\\0 & a\end{pmatrix}.  $
Block-diagonalizing $A$ by $W=RZ$, $R=\begin{pmatrix}1 & 1\\ 1 & -1
\end{pmatrix}$,
we obtain $Z'=(\tilde A+\tilde \Theta)Z$, where
$
\tilde A={\rm diag}(1,-1)$, 
$\tilde \Theta=R^{-1}\Theta R=
(a/2)\begin{pmatrix} 1& -1 \\
-1 & 1 \end{pmatrix}.
$
Applying the analog of \eqref{goodcondition}
on $\Re \lambda\ge 0$, 
where $\delta=2\Re \lambda^{1/2}\ge |\lambda|^{1/2}$, 
we obtain nonexistence of eigenvalues for 
$|\lambda|^{1/2}\ge 2|a|$, giving eigenbound
$|\lambda|\le 4|a|^2$.
By contrast, standard elliptic energy estimates 
give $|\lambda|\le |a|^2/4$,
which, by direct Fourier transform computation, is optimal.
Comparing, we see that the tracking bound is of the correct order,
$O(|a|^2)$, but with coefficient $16$ times larger than optimal.
\end{example}

This simple calculation 
may explain the ratio of roughly $10$ between 
the nonisentropic bounds found by tracking in Section \ref{hfnum}
and the isentropic energy bounds found in \cite{BHRZ,HLZ}.
The following result may be used to gauge at a practical level the 
efficiency of the analytical tracking bounds by a (nonrigorous, but typically
quite sharp) numerical convergence study.

\begin{proposition}\label{hflem}
On the nonstable half-plane $\Re \lambda \ge 0$, 
\begin{equation}\label{hflim}
\lim_{|\lambda|\to \infty}
D(\lambda)/ e^{
\alpha \lambda^{1/2}}=\hbox{\rm constant},
\end{equation}
\begin{equation}\label{alphaeq}
\alpha:=(1+\nu^{-1/2})\left(
\int_{-\infty}^0 (\hat v^{1/2}(x)-v_-^{1/2})\,\dif x
+\int_0^{+\infty} (\hat v^{1/2}(x)-v_+^{1/2})\,\dif x
\right) \, \hbox{\rm real}.
\end{equation}
\end{proposition}

\begin{proof}
Reviewing the proof of Proposition \ref{hf},
we find that the initial transformation $T$ is asymptotically constant
in $\lambda$ as $|\lambda|\to \infty$, and thus, the projection onto
the ``hyperbolic'' mode corresponding to the 1--1 entry has the
same property.
It follows that the Kato ODE $R'=(PP'-P'P)R$ used to propagate initializing
bases at $\pm \infty$ (see Section \ref{numprot}), 
%TODO: more precise ref. later.
where $P$ is the projection
onto the stable (respectively. unstable) subspace of $A_\pm$, asymptotically
decouples, yielding a constant (stable) hyperbolic basis element, and
two stable and two unstable ``parabolic'' basis elements coming from
the $4\times 4$ lower right-hand block of matrix $E$ further below.
But, the latter decouples into what may be recognized as a pair of
first-order systems corresponding to the scalar variable-coefficient
heat equations $u_t= u_{xx}$ and $u_t= \nu u_{xx}$.
Explicit evaluation of the Kato ODE, similar to but simpler than
the treatment of Burgers' equation in \cite{HLZ}, Appendix D,
then yields that the Evans function for would be asymptotically 
{\it constant} if $A(x,\lambda)$ were identically equal 
to $A_-$ for $x\le 0$ and $A_+$ for $x\ge 0$.
See also Remark \ref{bestevans}, which yields the same result
in much greater generality.

Though $A$ is not constant for $x\gtrless 0$, the $|\lambda|$-asymptotic
flow may be developed as in \cite{MZ.3} in an asymptotic series
in $\lambda^{-1/2}$ (respectively. $\lambda^{-1}$) in parabolic (respectively. hyperbolic)
modes, to see that, up to an asymptotically constant factor
(coming from $c(x)=O(1)$ terms in eigenvalues of various modes,
through $e^{\int_{\pm \infty}^0 (c(x)-c(\pm \infty))\,\dif x}$),
the flow from $y$ to $x$ is given asymptotically by 
$e^{\pm \lambda^{1/2}\int_y^x \nu^{-1/2} \hat v(z)^{1/2}\,\dif z}$
and
$e^{\pm \lambda^{1/2}\int_y^x \hat v(z)^{1/2}\,\dif z}$ in parabolic modes
and
$e^{-\lambda(y-x)}$
in the hyperbolic mode (note: constant rate, so no resulting
correction), 
whence, correcting for variation in the integrand
from the constant-coefficient case, we obtain \eqref{hflim}.

We omit the details, referring the reader to \cite{HLZ,CHNZ}
and especially \cite{MZ.3} for similar but more difficult calculations.
\end{proof}

\begin{remark}\label{fit}
We see from \eqref{hflim} that the asymptotic behavior of contours
is independent of shock amplitude or model parameters,
being determined up to rescaling of $\lambda$ by $\sgn \alpha$.
This explains the ``universal'' quality of contour diagrams arising
here and in \cite{HLZ,CHNZ}.
In practice, it is not necessary to compute $\alpha$,
since the knowledge that limit \eqref{hflim} exists allows
us to determine $\alpha$, $C$ by curve fitting of 
$\log D(\lambda)= \log C+ \alpha \lambda^{1/2}$ 
with respect to $z:=\lambda^{1/2}$, for $|\lambda|\gg 1.$
When $D$ is initialized in the standard way on the real axis,
so that $\bar D(\lambda)=D(\bar \lambda)$, $C$ is
necessarily real.
\end{remark}

\begin{remark}\label{winding}
Restricting the limiting Evans function $D^\dagger:=
Ce^{\alpha \lambda^{1/2}}$
to the imaginary axis, $\lambda=i\tau$, $\tau \in \mathbb{R}$,
we obtain
$$
D^\dagger(i\tau)=Ce^{\alpha |\tau|^{1/2}/2} (\cos + i\sin)(\pm \alpha 
|\tau|^{1/2}),
$$
predicting increasing winding about the origin as $\tau \to \infty$.
\end{remark}

Since the limiting Evans function $D^\dagger:=
Ce^{\alpha \lambda^{1/2} }$
is nonvanishing, we may obtain practical high-frequency bounds by a
convergence study on $D\to D^\dagger$, requiring, say, relative error
$\le 0.1$ to obtain a conservative but reasonably sharp radius.
Here, $D^\dagger$ may be estimated numerically using profile data
in the formulae of Proposition \ref{hflem}, by curve-fitting as described
in Remark \ref{fit}, or, more conventionally,
by numerical extrapolation in the course of the convergence study. 

See, for example, the computation displayed in Figures \ref{fig0a}--\ref{fig0b},
comparing the nonisentropic Evans function to its high-frequency
limit $ Ce^{\alpha \lambda^{1/2}} $
for $\Gamma=2/3$ and $\mu=\nu=1$, at contour radii $\Lambda=25$
and $\Lambda=10$, respectively, 
where $C$ and $\alpha$ have been determined by first taking limits along
the positive real axis.
Clearly, convergence in both cases has already occurred, whence 
radius $\Lambda=10$ is sufficient to bound unstable eigenvalues,
similarly as in the isentropic case \cite{BHRZ,HLZ}.
By comparison (see Section \ref{numprot}), tracking estimates
give the much more conservative bound $\Lambda=100.4$.
More extreme cases are depicted in Figures \ref{fig0c} and \ref{fig0d}
for $\mu=1$, $\nu=5$, and gas constants $\Gamma=2/3$ and $\Gamma=1/5$, respectively.
Note that convergence has already occurred at radius $\Lambda=40$,
which is thus sufficient to bound unstable eigenvalues; by contrast, the bounds obtained by tracking 
are $\Lambda= 391.3$ and $\Lambda= 1755.6$, respectively.

These figures clearly indicate the universal behavior of the high-frequency limit.
This can be seen also in the contour plots of Figure \ref{fig1},
comparing contours for the same model and radius at different values of $v_+$.

\begin{figure}[t]
\begin{center}
\includegraphics[width=10cm]{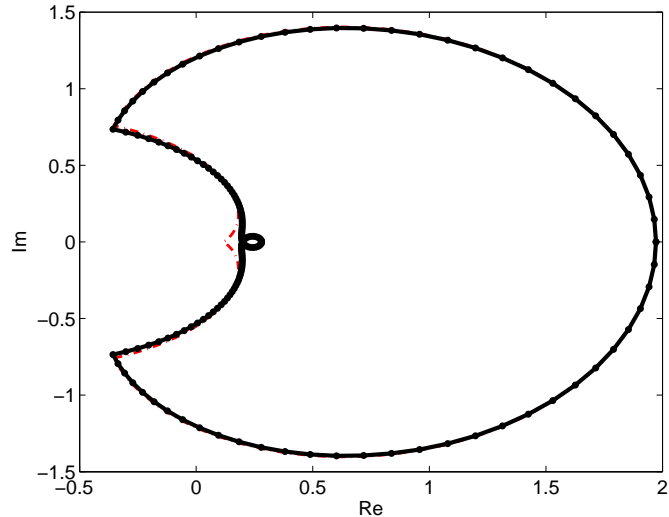} 
\end{center}
\caption{ Universal behavior at high frequency: The images of the semicircle of radius $25$ under the Evans function
and its universal approximant \eqref{genf} ($C$, $\alpha$, and $\beta$ determined by curve fitting), for a monatomic gas, $\Gamma=2/3$,
$\mu=\nu=1$, in the worst case $v_+=v_*=1/4$.  Agreement is nearly exact on the image of the outer,
circular arc and most of the imaginary axis, with deviations for $|\lambda|$ small.}
\label{fig0a}
\end{figure}

\begin{figure}[t]
\begin{center}
\includegraphics[width=10cm]{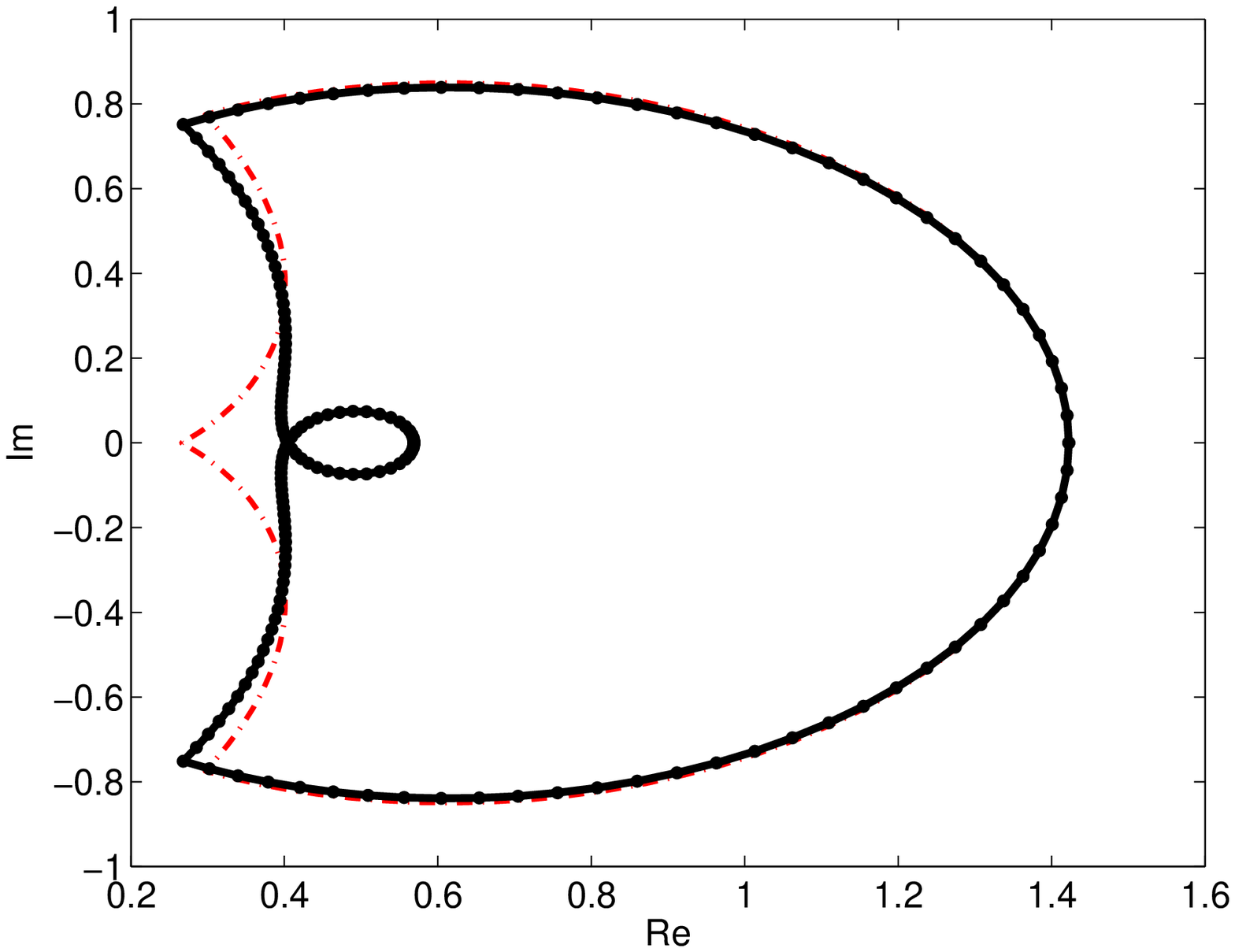} 
\end{center}
\caption{ Universal behavior: The images of the semicircle of radius $10$ under the Evans function
and its universal approximant \eqref{genf}, for $\Gamma=2/3$, $\mu=\nu=1$, $v_+=v_*=1/4$.
(Tracking radius $= 100.4$.)}
\label{fig0b}
\end{figure}

\begin{figure}[t]
\begin{center}
\includegraphics[width=10cm]{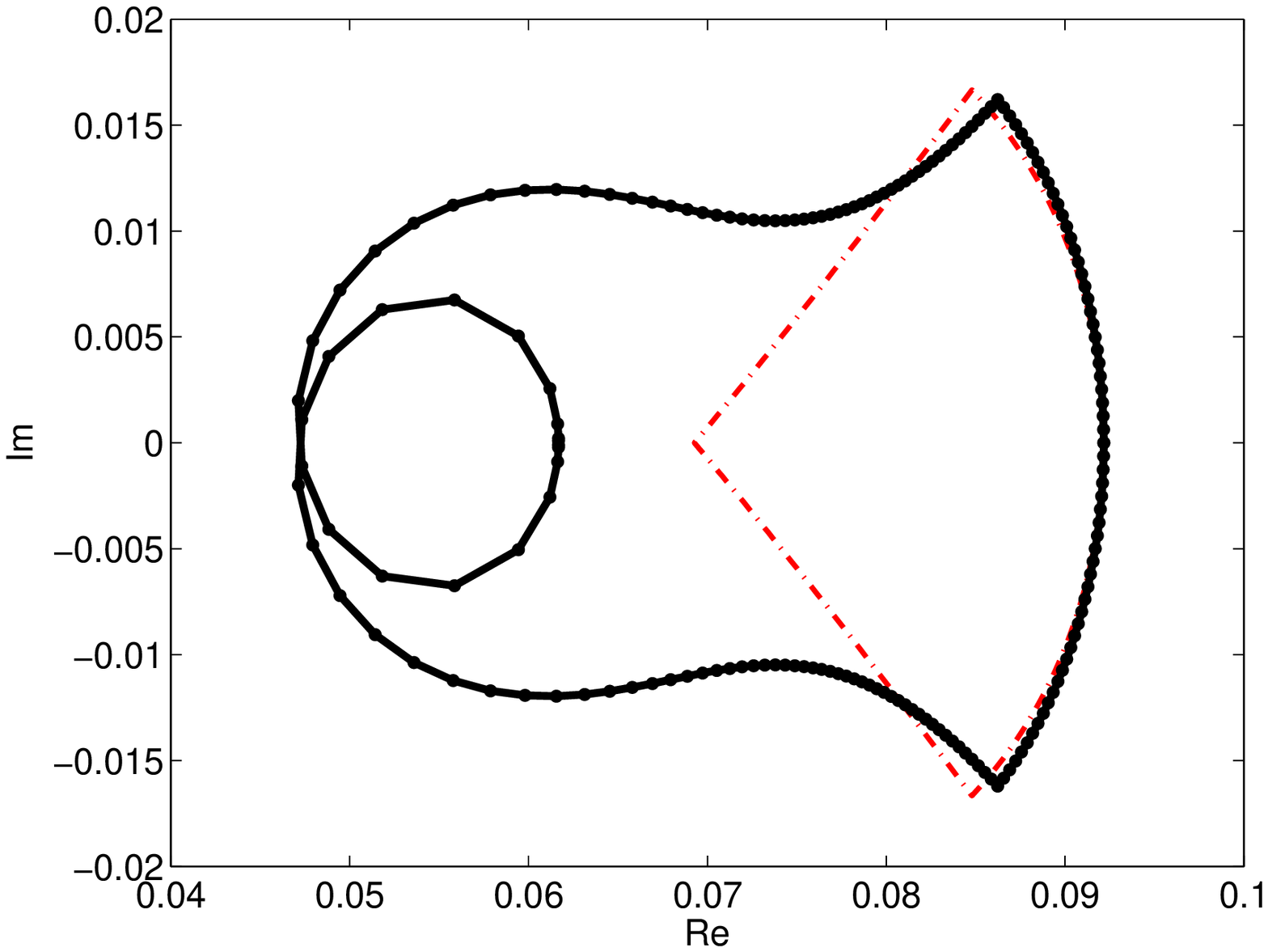} 
\end{center}
\caption{ Universal behavior: The images of the semicircle of radius $40$ under the Evans function
and its universal approximant \eqref{genf}, for $\Gamma=2/3$, $\mu=1$, $\nu=5$, $v_+=v_*=1/4$.
(Tracking radius $=391.3$.)}
\label{fig0c}
\end{figure}

\begin{figure}[t]
\begin{center}
\includegraphics[width=10cm]{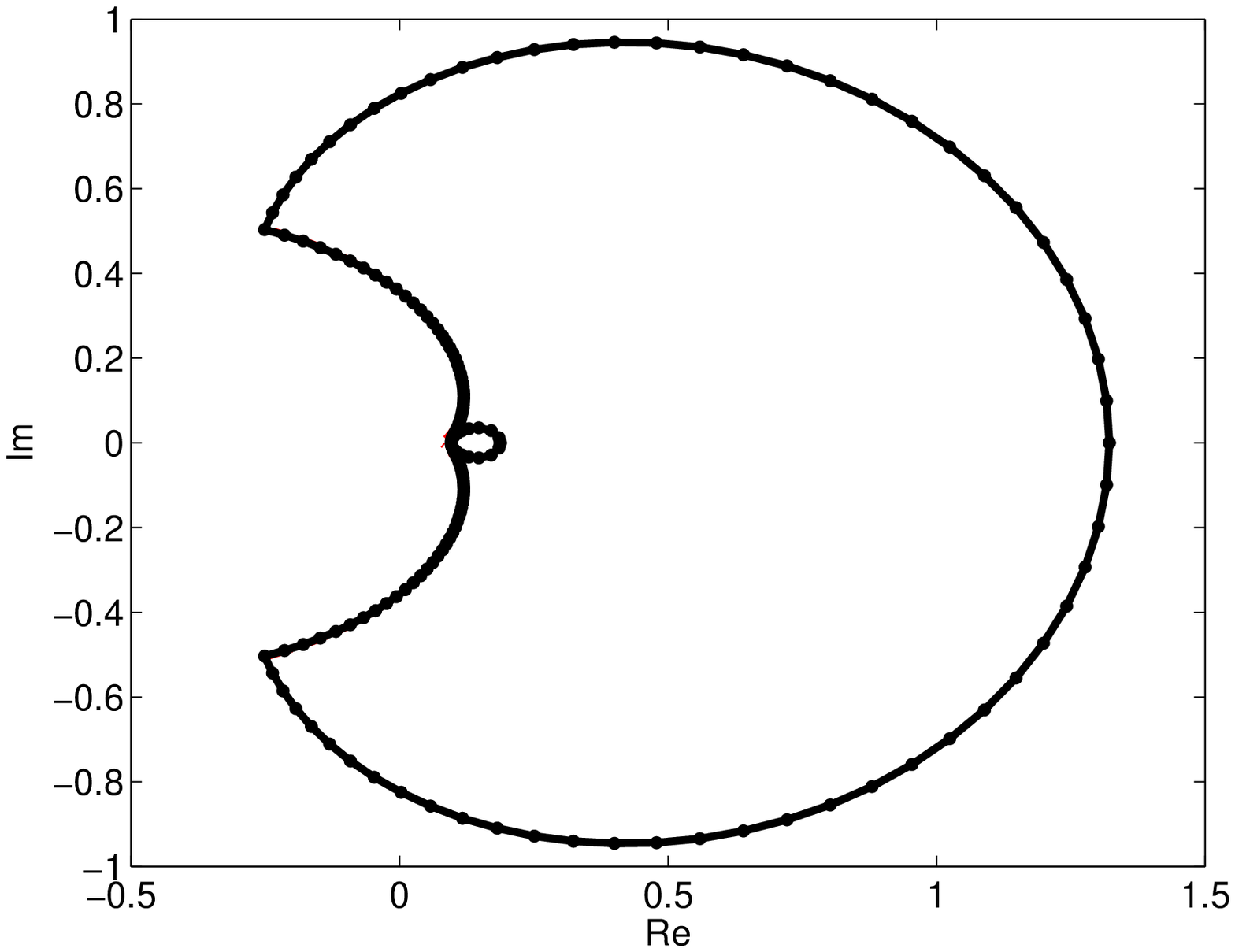} 
\end{center}
\caption{ Universal behavior: The images of the semicircle of radius $40$ under the Evans function
and its universal approximant \eqref{genf}, for $\Gamma=1/5$, $\mu=1$, $\nu=5$, $v_+=v_*=1/11$.
(Tracking radius $=1755.6$.).  Note that the two curves are essentially indistinguishable, except that the 
universal approximate does not loop but instead cusps near the origin.}
\label{fig0d}
\end{figure}

\begin{remark}\label{ghp}
More generally, the argument of Proposition \ref{hflim} yields
\begin{equation} \label{genf}
D(\lambda)/ e^{
\alpha \lambda^{1/2}+ \beta \lambda}=\hbox{\rm constant},
\end{equation}
$\alpha$ and $\beta$ real constants,
for all hyperbolic--parabolic parabolic system of the form
studied in \cite{MZ.3, MZ.4}, where $\beta$ corrects
for variation in the rates of growth (respectively. decay) in hyperbolic modes,
which are given to first order by $\lambda$ times 
their convection rates \cite{MZ.3}.
That $\beta=0$ in the present case is an accident of Lagrangian coordinates,
for which hyperbolic modes are convected (in the rest frame of the shock) 
with the constant fluid velocity $-s$.
In the more general case \eqref{genf}, 
asymptotic behavior of contours is determined
(up to rescaling in $\lambda$) by $\sgn \alpha$ together with the
additional parameter $\beta/\alpha^2$.
%TODO: was this supposed to be there? What does it mean?-K
% which would establish the conjecture.
\end{remark}

%TODO: resolve this question numerically, and either remove, modify,
%or strengthen the conjecture accordingly (KZ).
\begin{remark}\label{smallconj}
Figure \ref{fig0d} is particularly intriguing, showing convergence
also for small frequencies.  This suggests the conjecture
that $D^*$ might converge identically to
$D^\dagger$ in the singular limit $\Gamma \to 0$, $\nu/\mu\to \infty$
for all frequencies, small as well as large.
This would be an interesting direction for further investigation.
We remark that (i) the limit of $D^*$ as $\Gamma \to 0$ {\it is}
accessible by our techniques, Remark \ref{eblowup}, and
%(ii) The limiting case $\Gamma=0$, $v_+=0$ admits an exact profile
%solution $\hat v=\tanh (x/2)$, hence there is at least the possibility of
%an exact solution $D^*|_{\Gamma=0}(\lambda)\equiv Ce^{\alpha \lambda^{1/2}}$.
(ii) the limit $\nu/\mu\to \infty$ should be accessible by standard
singular perturbation techniques.
\end{remark}

%NOTE: This is also interesting in traveling pulse
%case treated by PeW, which coincides with our kato
%normalization.  Rather, something similar is true, I
%guess, not exactly this one.. TODO: include?  NO, not
%so relevant and getting to be too long anyway...

\subsection{The small-amplitude limit}\label{smallamp}
We mention in passing the following, related result
noted in \cite{HLZ}, regarding the {small-amplitude limit}
$v_+\to 1$.

\begin{proposition}\label{smallprop}
The Evans function $D$ converges uniformly as $v_+\to 1$
on compact subsets of $\{\Re \lambda \ge 0\}\setminus\{0\}$,
$\Gamma, \nu, \mu>0$ to a constant $C(\Gamma, \nu, \mu)$.
\end{proposition}

\begin{proof}
For $|\lambda|$ sufficiently large, this follows by
Proposition \ref{hflem} together with the fact \cite{P} that profiles
in the small-amplitude converge to an approximate
Burgers equation profile given by a symmetric $\tanh$ function,
for which $\alpha$ vanishes to order $|1-v_+|$ in \eqref{alphaeq}.
For $|\lambda|$ bounded, this follows as described in \cite{HLZ}
by the fact that the Evans function also converges in the small-amplitude
limit to the Evans function associated with the scalar, Burgers equation,
which by direct calculation is constant.
(The latter fact may be deduced, alternatively, by a simple scaling
argument showing that, for Burgers equation, the small-amplitude
limit and large-amplitude limits are equivalent.)
\end{proof}

%TODO: delete the below paragraph? 
%Not so relevant, but maybe nice orienting
%discussion for the reader.. (Jeff?  Greg?  What do you think?-K)
The significance of Proposition \ref{smallprop} 
%TODO: wording?-KZ
%comes from the fact
is
that the exponential rate of decay of profiles to their endstates $U_\pm$
goes to zero in the characteristic limit
$v_+\to 1$, as seen in the proof of Lemma \ref{profdecay}.
Thus, we cannot immediately conclude as in the ``regular'' 
large-amplitude limit $v_+\to v_*$ even that a limit 
exists as $v_+\to 1$.
Moreover, a second consequence is that the computational
domain $[-L_-,L_+]$ on which we carry out numerical evaluation
of the Evans function enlarges to $|L_\pm|\to \infty$
as $v_+\to 1$, since this must be taken roughly proportional
to the inverse of the exponential rate for numerical accuracy;
see Section \ref{inftys}.
Thus, the boundary $v_+=1$ is not directly accessible
by numerical Evans function techniques,
but requires a singular-perturbation analysis, either carried
out numerically or else analytically as above.
Alternatively, small-amplitude instability may be ruled out
by energy estimates as described in the introduction;
see \cite{MN, HuZ.2, BHRZ}.

\begin{remark}\label{smallamprmk}
Similarly as described just below Remark \ref{winding}
for the high-frequency limit, the small-amplitude limit
may be used to obtain a sharp but nonrigorous\footnote{
Recall that rigorous small-amplitude bounds are available by 
energy estimates \cite{HuZ.2,BHRZ}.}
lower bound on the amplitude of unstable shocks by
a convergence study requiring relative error $<1/2$
between Evans values and their constant limit: or,
still sharper, and more conveniently estimated, to 
require relative error $<1/2$ between the Evans function
and its estimated high-frequency approximant.
Convergence in the small- and large-amplitude limits
$v_+\to 1$ and $v_+\to v_*$ 
is illustrated numerically in Fig. \ref{fig1}.
\end{remark}

\section{Numerical protocol}\label{numprot}

We now describe the numerical algorithm, based on approximate
computation of the Evans function, by which we shall
locate any unstable eigenvalues, 
if they exist, in our system, over the compact parameter range under
investigation.  
Specifically, using analyticity of the Evans function in $\Re \lambda \ge 0$,
we numerically compute the winding number around a large semicircle
$B(0,\Lambda)\cap \{\Re \lambda \ge 0\}$ enclosing all possible unstable
roots, obtaining a count of the number of unstable eigenvalues within,
and thus of the total number of unstable eigenvalues.
This approach was first used by Evans and Feroe \cite{EF} 
and has been advanced further in various directions in,
for example, \cite{PW,PSW,AS,Br.1,Br.2,BDG,HuZ.2}.

\subsection{Approximation of the profile}\label{profnum}
Following \cite{BHRZ,HLZ}, we compute the traveling wave profile using \textsc{MatLab}'s {\tt bvp4c} routine, 
which is an adaptive Lobatto quadrature scheme.  These calculations are performed on a finite computational 
domain $[-L_-,L_+]$ with projective boundary conditions $M_\pm (U-U_\pm)=0$.
The values of approximate plus and minus spatial infinity $L_\pm$
are determined experimentally by the requirement that the
absolute error $|U(\pm L_\pm)-U_\pm|$ be 
within a prescribed tolerance $TOL=10^{-3}$.
This requirement is not too demanding in practice; we make more
stringent requirements later in evaluating the Evans function.

\subsection{Bounds on unstable eigenvalues}\label{hfnum} 

We next estimate numerically the coefficients
$|\cF^*_{kl}|(x, \Lambda)$ defined in \eqref{F*}, \eqref{Lambda},
using the numerically generated profiles described above,
generating an iterative sequence $\Lambda_{j+1}:=T(\Lambda_j)$,
$$
T(\Lambda):= 2\max\{1, \nu\} \max_{x} 
\Big( 
\frac{ |\cF^*_{--}|+|\cF^*_{++}| + 2\sqrt{|\cF^*_{-+}||\cF^*_{+-}|}}
{\hat v^{1/2}} 
(x, \Lambda)
\Big)^2.
$$
This is easily seen to converge, 
with odd terms monotone increasing and even terms monotone decreasing
or the reverse, to a fixed point $\Lambda_*=T(\Lambda_*)$,
which, by Proposition \ref{hf}, gives a bound 
$|\lambda|\le \Lambda_*$
on the maximum modulus of unstable eigenvalues $\Re \lambda\ge 0$.
Computations for a range of typical parameter values 
are displayed in Table \ref{tablerel2}, for the worst case
$v_+=v_*$.
Note the degradation of bounds for $\nu\gg\mu$ or $\nu\ll\mu$,
a consequence of multiple scales (stiffness).
% Rmk. \ref{trackrmk}.

\begin{table}
\begin{center}
\begin{tabular}{|c|ccccc|}
\hline
&\multicolumn{5}{c|}{$\nu$}\\
$\Gamma$ & 0.2 & 0.5 & 1.0 & 2.0 & 5.0\\
\hline
0.2 & 398.6 & 388.8 & 385.3 & 733.8 & 1755.6\\
0.4 & 211.7 & 182.3 & 175.1 & 325.0 & 762.0\\
0.6 & 222.5 & 123.4 & 111.5 & 198.4 & 449.8\\
2/3 & 226.8 & 114.3 & 100.4 & 175.3 & 391.3\\
0.8 & 236.5 & 103.9 & 85.3 & 142.6 & 307.1\\
1.0 & 253.8 & 100.8 & 73.7 & 113.8 & 229.2\\
1.2 & 274.3 & 106.6 & 69.7 & 98.7 & 183.1\\
1.4 & 300.8 & 117.5 & 70.5 & 91.6 & 154.9\\
1.6 & 347.4 & 131.8 & 74.5 & 89.8 & 138.0\\
1.8 & 397.3 & 148.7 & 80.8 & 91.8 & 128.6\\
2.0 & 450.3 & 167.5 & 88.7 & 96.7 & 124.7\\

\hline
\end{tabular}
\caption{Tracking bound $\Lambda_*$ vs. $\Gamma$, $\nu$ ($v_+=v_*$).}
\label{tablerel2}
\end{center}
\end{table}

\begin{remark}\label{trackrmk}
The poor rigorous tracking estimates obtained for $\nu\gg\mu$
or $\nu \ll\mu$ could in principle be improved by a refined tracking
estimate separating further the parabolic modes: that is, taking account
of the presence of multiple parabolic scales; see \cite{MZ.3, PZ}
for methodology.
%TODO: ideally?  (wording... -K)
Ultimately, this should be treated by singular perturbation techniques
as in \cite{AGJ},
separating out also fast/slow behavior of the background profile.
%In typical applications, $\nu/\mu \sim 1$; see Remark \ref{simplerule}.
\end{remark}

\subsubsection{High-frequency convergence study}\label{hfconstudy}
%Alternatively, we may obtain more realistic, but nonrigorous unstable
Alternatively, we could obtain more realistic, but nonrigorous unstable
eigenvalue bounds by a convergence study as $|\lambda| \to \infty$,
as described in Section \ref{hfconv}.
A convenient algorithm is to estimate coefficients $C$, $\alpha$ of
the high-frequency approximant $D(\lambda)\sim Ce^{\alpha \sqrt{\lambda}}$
by linear fit of $\log D\sim \log C + \alpha \sqrt{\lambda}$ as
$\lambda$ goes to real positive infinity, then 
approximate by binary search the value $\Lambda_*$ at 
which relative error between $D(\lambda)$
and $Ce^{\alpha \sqrt{\lambda}}$ is less than $TOL_1=10^{-1}$ on
the positive semicircle $|\lambda|=\Lambda_*$, $\Re \lambda\ge 0$,
indicating convergence to this tolerance, hence nonvanishing
of $D$, for $|\lambda|\ge \Lambda_*$, $\Re \lambda\ge 0$.
%TODO: restore later if data available???
%The resulting bounds, displayed in Table \ref{tablerel4},
%are much less conservative than those obtained by rigorous tracking
%estimates of Section \ref{hfnum}.
%(TODO: put in data-Jeff)
%CONJECTURE: Should give bound $\sim 10$ as in isentropic case;
%see Example \ref{trackeg}.
The resulting bounds
are much less conservative than those obtained by rigorous tracking
estimates of Section \ref{hfnum};
see for example the discussion following Remark \ref{winding}.
We do not pursue this here, as our main interest is in rigorous bounds.
{However, the observation seems quite important for 
practical numerical testing,}
as typical differences in radius are an order of magnitude or more, and
the size of the radius appears to be the main limiting factor
in computational efficiency.

\subsection{Approximation of the Evans function}\label{evansnum} 
We compute the Evans functions for comparison by two rather different 
techniques,
both of which have been demonstrated to give good numerical results.

\subsubsection{The exterior product method}
Following \cite{Br.1, Br.2, BrZ, BDG}, we work with
``lifted equations''
$$
\cW'=A^{(k)}\cW, \quad \cW:=W_1\wedge \cdots \wedge W_k,
$$
evolving subspaces encoded as exterior products of basis
elements $W_j$, where
\begin{equation}\label{liftdef}
A^{(k)} (W_1\wedge \cdots \wedge W_k):=
(AW_1\wedge \cdots \wedge W_k)+ \dots +
(W_1\wedge \cdots \wedge AW_k),
\end{equation}
defining $\cW^+$ and $\cW^-$ as $k_+$- and $k_-$-products
of bases $\{W_j^+\}$ and $\{W_j^-\}$ 
of the subspaces of solutions of $W'=AW$ decaying
at $+\infty$ and $-\infty$;
in the present case, $k_+=3$, $k_-=2$.

In this setting, the stable (respectively. unstable) manifold at $+\infty$
(respectively. $-\infty$) corresponds to a single solution/vector, eliminating
difficulties of ``parasitic modes'', etc.; see \cite{BrZ,HuZ.2,BDG,BHRZ}
for further discussion.
We then compute the Evans function as
$$
D(\lambda)= \cW^+\wedge\cW^-|_{x=0}
$$
or, alternatively, as
$
D(\lambda)= \tilde \cW^+\cdot \cW^-|_{x=0},
$
where $\tilde \cW^+$ is an appropriate solution of
the adjoint equation 
$
\tilde \cW'=-(A^{(k)})^*\tilde \cW;
$
see \cite{Br.1, Br.2, BrZ, BDG, HuZ.2,BHRZ} for further details.

\subsubsection{The polar coordinate method (``analytic orthogonalization'')}
An alternative method proposed in \cite{HuZ.2} is to encode
$\cW=\gamma \Omega$, where ``angle''
$\Omega=\omega_1\wedge \cdots \wedge \omega_k$
is the exterior product of an orthonormal basis $\{\omega_j\}$
of $\Span \{W_1, \dots, W_k\}$ evolving independently
of $\gamma$ by
some implementation (e.g., Davey's method) of continuous orthogonalization
and ``radius'' $\gamma$ is a complex scalar evolving by a scalar ODE
slaved to $\Omega$, related to Abel's formula for evolution of a full
Wronskian; see \cite{HuZ.2} for further details.
This might be called ``analytic orthogonalization'', as the main difference
from standard continuous orthogonalization routines is that it restores
the important property of analyticity of the Evans function by the introduction
of the radial function $\gamma$ ($\Omega$ by itself is not analytic).

{\bf Advantages/disadvantages:}  The exterior product method is
{\it linear}, but evolving in a high-dimensional ($\sim n^n$) space.  
The polar coordinate method is {\it nonlinear}, hence less well-conditioned
and involving more complicated function calls, but evolving on
a lower-dimensional manifold ($\sim n^2$).
Thus, there is a tradeoff in dimensions vs. conditioning, with the polar
coordinate method the only reasonable option for high-dimensional systems
and the exterior product method somewhat faster and more efficient for
low-dimensional systems \cite{HuZ.2}.  Both methods are effective (and
reasonably comparable in efficiency) for $n\le 7$ or so.

{\bf Role as numerical check:} 
Since the two methods involve completely different
ODE, one linear and the other nonlinear, agreement in their results
is strong, if indirect, evidence that equations have been 
properly encoded and solutions accurately approximated,
{\it at least on the finite computational domain $[-L_-,L_+]$}.
We address in the following subsection
the separate question of determining appropriate $L_\pm$.

\subsubsection{Determination of approximate spatial infinities}\label{inftys}
%TODO: improve this part- maybe put in sharp bound then relax later?
Denoting by $A^{(k)}_\pm(\lambda)$ the limits at $\pm \infty$ of the lifted
matrix $A^{(k_\pm)}(x,\lambda)$ defined in \eqref{liftdef},
and $\mu_+$ (respectively. $\mu_-$ the eigenvalue of $A^{(k)}_+$ (respectively. $A^{(k)}_-$) 
of smallest (respectively. largest) real part, we find that there holds
a uniform bound
\begin{equation}\label{crude}
e^{(A^{(k)}-\mu)_\pm x}\le C_*, \quad x\lessgtr 0
\end{equation}
on any compact subset of $\Re \lambda\ge 0$,
for $\Gamma$ bounded from zero, and $\mu$, $\nu$ bounded
and bounded from zero, for some $C_*>0$.
For $\lambda$ bounded from zero,
this follows by consistent splitting on $\{\Re \lambda \ge 0\}\setminus\{0\}$, 
and the choice of $k_\pm$ as dimensions of stable (respectively. unstable) 
subspaces of $A_\pm$, which together imply that, away from $\lambda=0$,
$\mu_\pm$ are simple eigenvalues of $A^{(k)}_\pm$.
For $\lambda$ near zero, on the other hand, we may verify directly
that $\mu_\pm$ are semisimple, by the same considerations used
to verify continuous extension of stable (respectively, unstable) subspaces
of $A_\pm$:
simple in the case of $\mu_-$, since the unstable subspace remains
uniformly spectrally separated from remaining eigenvalues of $A_-$; 
semisimple in the case of $\mu_+$,
because hyperbolic characteristics $a_j^+$ are simple, and thus eigenvalues
$\mu_j^+$ of small real part, 
by the standard theory of \cite{GZ,ZH,MZ.3}, are analytic and semisimple,
of form $\mu_j^+\sim -\lambda/a_j^+$,
whence $\mu^+$ (the sum of the $k_+=3$ eigenvalues of largest real part)
is semisimple as well.

%By the theoretical development of \cite{Br.1, Br.2, BHRZ, HLZ},
%specifically, the ``quantitative gap lemma'' of Theorem C.2, \cite{BHRZ},
Applying the ``quantitative gap lemma'' of Theorem C.2, \cite{BHRZ},
we have therefore that the relative error
% $\cE$ 
between the solution $\cW^\pm(\pm L_\pm)$
at plus or minus approximate spatial infinity $x=\pm L_\pm$ and
the constant-coefficient initialization $\cV^\pm e^{\mu_\pm \pm L_\pm}$
is bounded by $\frac{\epsilon}{1-\epsilon}$, for
\begin{equation}\label{quantgap}
\epsilon:= C_*|A^{(k)}(\cdot, \lambda)-A^{(k)}_\pm 
(\lambda)|_{L^1(\pm L,\pm \infty)}.
\end{equation}
Using the bound $|M^{(k)}|\le k|M|$ established in Appendix \ref{liftbd},
and the asymptotic behavior
\begin{equation}\label{asymptotics}
A_j(x,\lambda)-A_{j,\pm} (\lambda)\approx Q_je^{-\theta_\pm |x|},
\end{equation}
$A=:\lambda A_1 + A_0$ for $x\to \pm \infty$, we may thus estimate
$$
\begin{aligned}
\max_{|\lambda|\le \Lambda}
&|A^{(k)}(\cdot, \lambda)-A^{(k)}_\pm (\lambda)|_{L^1(\pm L,\pm \infty)}
\le 
\max_{|\lambda|\le \Lambda}
k |A(\cdot, \lambda)-A_\pm (\lambda)|_{L^1(\pm L,\pm \infty)}\\
&\approx 
\max_{|\lambda|\le \Lambda}
 \frac{k |A(\pm L_\pm, \lambda)-A_\pm (\lambda)|}{\theta_\pm}\\
&\le
  k\frac{ |A_0(\pm L_\pm)-A_{0,\pm} |+ \Lambda
 |A_1(\pm L_\pm)-A_{1,\pm} | }{\theta_\pm}\\
&=
  k\frac{ |A(\pm L_\pm,0)-A_{\pm}(0) |} {\theta_\pm}\\
&\quad +k\frac{ |A(\pm L_\pm, \Lambda)-A_{\pm}(\Lambda)-
 (A(\pm L_\pm,0)-A_{\pm}(0) )| }{\theta_\pm}.\\
\end{aligned}
$$

This gives a theoretical relative error bound
of $TOL$ (tolerance) between initializing basis at $\pm L$ and 
actual basis for the theoretical Evans function, 
provided that
$$
\begin{aligned}
C_*k&\Big( \frac{ |A(\pm L_\pm,0)-A_{\pm}(0) |} {\theta_\pm}
+\frac{ |A(\pm L_\pm, \Lambda)-A_{\pm}(\Lambda)-
 (A(\pm L_\pm,0)-A_{\pm}(0) )| }{\theta_\pm}\Big)\\
&
\le \frac{TOL}{1+TOL}\approx TOL,
\end{aligned}
$$
or, approximately,
\begin{equation}\label{Lcond}
\begin{aligned}
 &|A(\pm L_\pm,0)-A_{\pm}(0) |
+ |A(\pm L_\pm, \Lambda)-A_{\pm}(\Lambda)-
 (A(\pm L_\pm,0)-A_{\pm}(0) )|\\
&\quad \le
\frac{\theta_\pm}{C_*k} \, TOL.
\end{aligned}
\end{equation}

\begin{remark}\label{Lform}
Alternatively, working directly from \eqref{asymptotics},
we may solve \eqref{Lcond} with equality for
\begin{equation}\label{Lest}
L_\pm \approx \frac{\log C_* + \log k+ 
\log (|Q_0^\pm| + \Lambda |Q_1^\pm|) + \log \theta_\pm^{-1}
+\log TOL^{-1}}{\theta_\pm}.
\end{equation}
A reasonably good bound (noting insensitivity of log, and also
cancellation in large $\lambda$ vs. small $\theta$ effects) for $k=2$,
$\Lambda \sim 10^2$, $TOL=10^{-3}$, $|Q_1|=1$ for
the sparse matrix $Q_1$, $C_*=10^2$, 
and throwing out $\theta$ and $|Q_0|$ terms
as negligible for large $|\lambda|$, is
$$
L_\pm\approx \frac{\log 2 + 7\log 10  }{\theta_\pm}
\approx \frac{17}{\theta_\pm}.
$$
%(TODO: Jeff, please compare results of this vs. exact above.- K)
\end{remark}

%\begin{remark}\label{Lestrmk}
%(adapted from \cite{BHRZ}. TODO: update to present situation.)
%Applying \eqref{Lest}, we obtain bounds for 
%$TOL=10^{-3}$ for $L_+\sim TODO$. 
%This is conservative, as we have made little effort to optimize bounds, but still within the realm of our experiments.  Our numerical convergence studies 
%indicate that $L_\pm=TODO$ in fact suffices for $10^{-3}$ accuracy.
%\end{remark}

{\bf Numerical algorithm}
Starting with the values needed for accurate profile approximation,
we increment $L_\pm$ by some fixed step-size (here, we have chosen
step-size $5$) until \eqref{Lcond} is satisfied, 
taking $k=2$, $TOL=10^{-3}$, and
conservatively estimating $C_*=10^2$.
In Table \ref{tablerel4}, we have displayed values of
$\theta_\pm$, $L_\pm$ computed from \eqref{Lcond} with $TOL=10^{-3}$,
$C_*=10^2$, 
for various values of $\Gamma$ and $v_+$, with $\mu=1$ and $\nu$ 
set to the value $\nu= (3/4)\frac{9\gamma -5}{4}$,
$\gamma=\Gamma+1$, predicted by the kinetic theory of gases;
see Appendix \ref{simple}.
In principle, $C_*$ could be estimated numerically for a more
precise bound.
In practice, convergence studies reveal these bounds to be
rather conservative; see
Table \ref{tablerel5}, or \cite{BHRZ,HLZ} in the isentropic case.
%TODO: or, current paper, somewhere, I hope! (TODO- Jeff)
%(TODO: or, just state that we determine $L_\pm$ this way
%for all runs?  Depends I guess on outcome of tests...-K
%JEFF- please update the final statements explaining what we
%do, depending on what exactly you did.-K)

%TODO: Here, there should be convergence study testing
%actual rel. error (measured by larger L solution) at
%L_\pm and again at x=0..

\begin{table}
\begin{center}
\begin{tabular}{|c|ccc|ccc|}
\hline
\multicolumn{7}{|c|}{$v_+=0.7$}\\
\hline
$\Gamma$ & $L_-$ & $\theta_-$ & $17/\theta_-$ & $L_+$ & $\theta_+$ &
$17/\theta_+$\\
\hline
0.2 & 68 & 0.28 & 61 & 68 & 0.27 & 61\\
0.4 & 75 & 0.26 & 67 & 74 & 0.26 & 67\\
0.6 & 83 & 0.23 & 74 & 81 & 0.23 & 74\\
2/3 & 85 & 0.22 & 76 & 83 & 0.23 & 76\\
0.8 & 90 & 0.21 & 81 & 87 & 0.22 & 81\\
1.0 & 98 & 0.20 & 87 & 92 & 0.21 & 87\\
\hline
\multicolumn{7}{|c|}{$v_+=v_*$}\\
\hline
$\Gamma$ & $L_-$ & $\theta_-$ & $17/\theta_-$ & $L_+$ & $\theta_+$ &
$17/\theta_+$\\
\hline
0.2 & 22 & 0.91 & 19 & 24 & 0.83 & 19\\
0.4 & 28 & 0.70 & 25 & 27 & 0.71 & 25\\
0.6 & 34 & 0.57 & 30 & 31 & 0.61 & 30\\
2/3 & 36 & 0.53 & 32 & 33 & 0.59 & 32\\
0.8 & 41 & 0.48 & 36 & 36 & 0.54 & 36\\
1.0 & 47 & 0.41 & 42 & 40 & 0.48 & 42\\
\hline
\end{tabular}
\caption{$L_+$($\theta_+)$/$L_-$($\theta_-$) vs. $\Gamma$, $v_+$.}
\label{tablerel4}
\end{center}
\end{table}

{\bf Translation to $x=0$.}
The above-described algorithm is designed to achieve relative accuracy
$TOL$ of $\cW_\pm$ at $x=\pm L_\pm$, whereas the 
accuracy of the Evans function is determined, rather, by their
relative errors at $x=0$.
A complete description of the error must thus include also a discussion
of possible magnification through the evolution from $\pm L_\pm$ to $0$.
However, as discussed in \cite{Br.1,Br.2}, the flow toward $x=0$
is in fact quite stable, since, by construction, the modes $\cW_\pm$
we seek to approximate are the fastest decaying in the flow toward
infinity, hence the fastest growing in backward flow toward zero, 
with all other modes decaying exponentially in relative size.
Thus, in practice, the resulting magnification in error is negligible;
see \cite{Br.1,Br.2} for further discussion or numerical studies.

%TODO (Jeff): convergence study for some representative value(s),
%to back all of this up? 

%TODO: restore later? (NO, mostly done anyway...--K)
%\begin{example}\label{Leg}
%TODO: for typical gas (say, diatomic? monatomic?), compute
%relative errors of $\cW^\pm$ at $L_\pm$ and at $x=0$, by 
%comparison with ``numerical exact solution'' obtained for
%larger value of $L$.  This should validate above observations
%by (i) exhibiting rel. error at $L_\pm$ of $\le TOL:=10^{-2}$,
%and (ii) exhibiting similar or smaller rel. error at $x=0$,
%also $\le TOL$.  Finally, compute rel. error of $D$ to see
%how it all comes out.
%\end{example}

\begin{remark}\label{bothL}
Since the polar coordinate method computes the
same quantity $\cW$ in different coordinates,
the initialization error bounds derived for the exterior
product method apply also for the polar coordinate method and 
so the same values $L_\pm$ may be used for both computations.
\end{remark}

\subsubsection{ODE protocol}
Following \cite{Br.1,Br.2,BrZ,BDG,HuZ.2,BHRZ},  
to further improve the numerical efficiency and accuracy of the 
shooting scheme, we rescale $\cW$ and $\widetilde{\cW}$ 
to remove exponential growth/decay at infinity, and thus eliminate potential problems with stiffness.  
Specifically, we let $\cW(x) = e^{\mu^- x} \cV(x)$, where $\mu^-$ is the growth rate of the unstable manifold at $x=-\infty$, and we solve instead 
$$
\cV'(x) = (A^{(k)}(x,\lambda)-\mu^- I)\cV(x).
$$
We initialize $\cV(x)$ at $x=-\infty$ to be the eigenvector of 
$A^{(k)}_-(\lambda)$ corresponding to $\mu^-$.  
Similarly, it is straightforward to rescale and initialize 
$\widetilde{\cW}(x)$ at $x=+\infty$.
To produce analytically varying Evans function output, the 
initial data $\cV(-L_-)$ and $\widetilde{\cV}(L_+)$ 
must be chosen analytically using \eqref{kato}.
The algorithm of \cite{BrZ} works well for this purpose, as discussed
further in \cite{BHRZ, HuZ.2}.

The ODE calculations for individual $\lambda$ are
carried out using \textsc{MatLab}'s {\tt ode45} routine, 
which is the adaptive 4th-order Runge-Kutta-Fehlberg method (RKF45), 
after scaling out the exponential decay rate of the constant-coefficient 
solution at spatial infinity, as described just above.
This method is known to have excellent accuracy \cite{BDG,HuZ.2};  
in addition, the adaptive refinement gives automatic error control.  
Typical runs involved roughly $300$ mesh points per side, with error tolerance set to {\tt AbsTol = 1e-6} and {\tt RelTol = 1e-8}.  
%TODO: Jeff- please update the final sentence to current values..-K

\begin{table}
\begin{center}
\begin{tabular}{|c|rcc|}
\hline
$L$ & $\mbox{rel}(\widetilde{W}^+(0))$ & $\mbox{rel}(W^-(0))$ &
$\mbox{rel}(D)$\\
\hline
20 & 2.2(-2) & 5.8(-3) & 3.8(-3)\\
25 & 3.3(-3) & 8.8(-4) & 4.6(-4)\\
30 & 4.7(-4) & 1.3(-4) & 5.3(-5)\\
35 & 6.4(-5) & 1.8(-5) & 5.7(-6)\\
40 & 8.7(-6) & 2.6(-6) & 8.9(-7)\\
45 & 1.1(-6) & 3.5(-7) & 2.2(-7)\\
\hline
\end{tabular}
\caption{
%Shows the convergence 
Convergence of $\widetilde{W}^+(0)$, $W^-(0)$, and
$D$ as $L$ increases from $25$ to $50$, incrementing by $5$, for a
diatomic gas ($\gamma=7/5$ and $\nu/\mu=1.47$.  The values were
computed for a quarter circle of radius $R=69$ consisting of 90
points.  Relative errors were computed at each point and the maximum
value along the contour is reported with the next higher value of $L$
being the baseline.}
\label{tablerel5}
\end{center}
\end{table}

\subsection{Winding number computation}\label{windingalg}

We compute the winding number for fixed parameter values
about the semicircle $\partial\{\lambda:\,
\Re \lambda\ge 0, \, |\lambda|\le \Lambda\}$ 
%TODO: update!
by varying values of $\lambda$ along $180$ points of the contour,
with mesh size taken quadratic in modulus to concentrate
sample points near the origin where angles change more quickly,
and summing the resulting changes in ${\rm arg}(D(\lambda))$,
using $\Im \log D(\lambda) ={\rm arg} D(\lambda) ({\rm mod} 2\pi)$,
available in \textsc{MatLab} by direct function calls.
As a check on winding number accuracy, we test a posteriori that the 
change in argument of $D$ for each step is less than $\pi/8$.  
%TODO: $\pi/8$? or $\pi/25$ as before??
Recall, by Rouch\'e's Theorem, that accuracy is preserved so long 
as the argument varies by less than $\pi$ along each mesh interval.

\section{Numerical results}\label{numresults}

Finally, we describe our numerical results in various cases,
using the numerical method just described, varying $v_+$
between $.7$ and the theoretical lower value $v_*$ in
our rescaled coordinates \eqref{scaling}.
For comparison between these and standard coordinates, it is useful
to convert these values to their associated {\it Mach number}, a
standard dimensionless measure of shock strength discussed further 
in Appendix \ref{mach}. As computed in \eqref{machform}, this
is given by
$ M=\frac{\sqrt{2}}{\sqrt{(\Gamma +2)(v_+-v_*)}} $,
which for $1-v_+$ small reduces using \eqref{phys} to
$$
M=\frac{1}{\sqrt{1-   (1-v_+)\frac{2+\Gamma}{2}} }
\approx 
1 +  (1-v_+)\frac{2+\Gamma}{4} \le 1 + |1-v_+|
$$
for the range $\Gamma\in [0,2]$ under consideration.
In particular, for the upper limit $v_+=.7$ of our computations, 
we have on the reduced range $0\le \Gamma\le 1$ the exact upper bound 
$ M\le (.55)^{-1/2} $, or approximately Mach $1.35$,
%CHANGED: we don't know whether it is in the range proved
%analytically, just that it is pretty darn small.... K
well within the small-amplitude regime. 
%for which stability has been shown analytically in \cite{HuZ.1}.
Recall that stability of small-amplitude shocks
has been shown analytically in \cite{HuZ.1}.
%ENDCHANGED
%TODO: This is a weak point I think... Really, we should go
%down a bit smaller I guess to be sure...
%SOMETHING TO ADD IN REVISION?

The smallest computed physical value $v_+-v_*=10^{-3}$ corresponds to
Mach $M\sim 100$, at which we see already convergence of the Evans
function to that of the nonphysical limit $v_+=v_*$, corresponding
to Mach $M=\infty$.
For a visual comparison, we display iso-Mach (constant Mach number) 
curves in the $\Gamma$-$v_+$ plane in Figure \ref{machfig}.

\begin{figure}[t]
\begin{center}
\includegraphics[width=10cm]{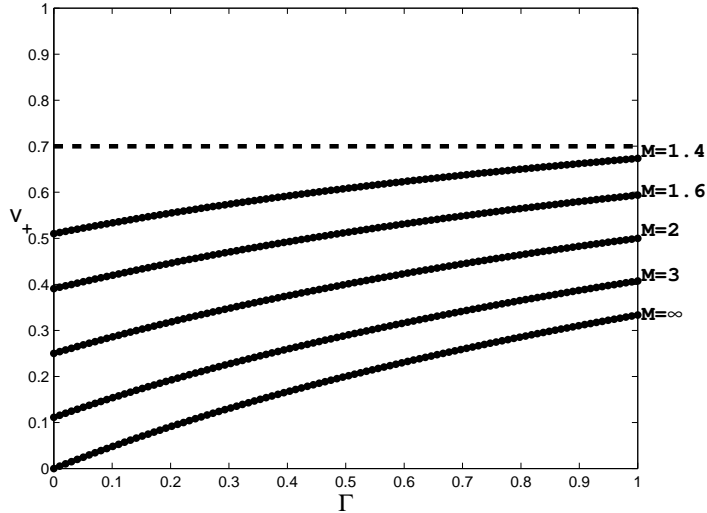}
\end{center}
\caption{ Iso-Mach curves in $\Gamma$, $v_+$.}
\label{machfig}
\end{figure}

%\subsection{Results for $\nu/\mu\sim 1$}\label{windingnum} 

%(Main case)
%
%TODO: Display of convergence, computation
%of winding numbers, etc. (Jeff).
%TODO? Comment: qualitatively extremely similar to the
%contours of the isentropic case; cf. \cite{HLZ}. NO!
%Cons. of universal hf behavior rather... KZ
%
%NOTES:

\subsection{Description of experiments: broad range}\label{windingnum} 
In the main case considered, of $\mu$ and $\nu$ of comparable 
but wide-ranging
size, a first set of experiments was carried out in the
range $\Gamma\in [.2,2]$, $\nu\in [0.2,5]$, sampling at mesh points
\[
\begin{aligned}
(\Gamma,\nu) &\in \{0.2,0.4.,0.6,2/3,0.8,1.0,1.2,1.4,1.6,1.8,2.0\}\\
& \times \{0.2,0.5,1.0,2.0,5.0\},\\
\end{aligned}
\]
$55$ pairs in all, where, for each value $(\Gamma, \nu)$, $v_+$
was varied among $8$ mesh points on a logarithmic scale from $.7$ to $v_*$,
for a total of $440$ runs in all.
In each of the cases that we examined, the winding number was zero,  
indicating spectral stability 
and thereby nonlinear stability and existence of shock layers 
in the vanishing viscosity limit, by the results of 
\cite{MZ.4,Z.3,GMWZ.1,GMWZ.2,GMWZ.3}.
These runs took $12$ days to complete, of which 8 days were spent
on the upper extreme case $\nu=5$, and 
%$TODO$  FIX THIS DISCUSSION
$2$ days were spent on the
lower extreme case $\nu=.2$, both out of physical range.

\subsection{Description of experiments: physical range}
A second set of experiments was carried out for $\Gamma$ values 
corresponding to monatomic, diatomic, etc. gas on a refined mesh
for $\nu$ within the smaller, physical range $\nu\in [1,2]$
indicated by Appendices \ref{constants} and \ref{simple},
sampling at
\[
\begin{aligned}
(\Gamma,\nu) &\in \{2/11,2/9,2/7,2/5,2/3\}\\
& \times \{1.0,1.1,1.2,1.3,1.4,1.5,1.6,1.7,1.8,1.9,2.0\},\\
\end{aligned}
\]
with $v_+$ again
varied among $8$ mesh points on a logarithmic scale from $.7$ to $v_*$,
again a total of $440$ runs.
The results again were winding number zero for each case tested,
indicating spectral stability.  These runs took 10 days to complete.
We remark that runs for $\nu=5$ and $\nu=.2$ took over twice as long to complete 
compared with the average, again reflecting the stiffness alluded to
in Remark \ref{trackrmk}, associated with difference in scale between $\nu$
and $\mu=1$.
%This is really inefficient, already almost out of reasonable range with rigorous
%tracking bounds.
%The bulk of the remaining time was taken up near the characteristic
%small-amplitude limit, with running time typically monotone decreasing
%as shock strength increased.
%Likewise, run-time was seen to decrease with increasing $\Gamma$.
%As mentioned in Section \ref{hfconstudy}, run time could be decreased
%by an order of magnitude or more
%by subsituting for the radius bounds obtained by rigorous tracking estimates
%the more accurate but nonrigorous bounds obtained by convergence study 
%toward the high-frequency limit.

\begin{figure}[t]
\begin{center}
\includegraphics[width=10cm]{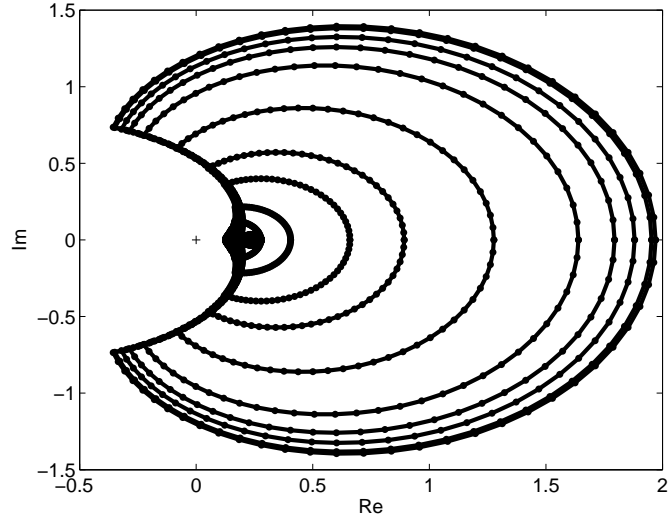} 
\end{center}
\caption{ Convergence to the limiting Evans function as $v_+\to v_*$ for a monatomic gas, $\Gamma=2/3$, $\mu=\nu=1$, $v_*=1/4$.
The contours depicted, going from inner to outer, are images of the semicircle of radius $25$ under the Evans function $D$
for $v_+=.7,.6,.5,.4,.35,.3,.27,.26,.255,.251,.2501$.  The outermost contour is the image under the limiting Evans function $D^*$,
which is essentially indistinguishable from the images for $v_+=.251$ and $v_+=.2501$.}
\label{fig1}
\end{figure}
%The values of v+ are .5,.3,.26,.251,.25 (monatomic v_star=.25).
%
%Radius=25, nu=1.

%OLD, omitted:
%\subsection{Results for $\nu/\mu>>1$, $\nu/\mu<<1$}\label{noncomparable} 
%In this section, use sharper but nonrigorous bounds on unstable
%eigenvalues obtained by convergence study of high-freq. limit.
%(Since otherwise out of numerical range).
%We find again stability.
%
%NOTE: for typical gases, $\nu/\mu\sim 1$ (Remark \ref{simplerule}).
%For liquids, may be very much smaller- however, ideal gas assumptions
%break down anyway, as discussed in Appendices.. (K).
%
%TODO: fill in.
%
%%TODO: maybe later this? (optional)
%\subsection{Monatomic and diatomic gas}\label{monadiatomic} 
%TODO: in this section do two main physical cases with theoretical $\nu$
%value in somewhat greater detail, with convergence studies, etc.
%

%TODO: restore later?  So far, inconclusive, not sure enough
%to include.... LEAVE FOR LATER... demote to remark in appendix...
%\subsection{Temperature-dependent case}\label{tdepcase} 
%TODO: full run, $v_+$ going from $0.7$ down to $v_*$, for
%monatomic case with Chapman law $q=1/2$ and diatomic with
%empirical (or?) law $q=.76$.
%
%%(great for physical side!)
%
%If numerical convergence (as I conjecture), very intriguing/suggestive
%for further investigation/analytic demonstration using analysis
%like that of \cite{HLZ}.

\appendix

\section{Gas constants for air}\label{constants}

In this appendix, we list the various relevant gas constants
for dry air at normal temperatures and pressures.
For the specific
%universal 
%NO! specific gas constant equals universal constant divided by
%molar weight!
gas constant (universal gas constant $R\approx 8.3\,\frac{\unit{J}}{{\rm moles} 
\cdot \unit{K}}$ 
divided by molar mass), 
we have
$$
\bar R\approx 287.05\, \frac{\unit{J}}{\unit{kg}\cdot\unit{K}},
$$
$\unit{J}=$ Joules, $\unit{kg}=$ Kilograms, $\unit{K}=$ degrees Kelvin \cite{Bro2}.
%\footnote{http://en.wikipedia.org/wiki/Density\_of\_air}
For specific heat at constant volume, we have
$$
c_v\approx  716\, \frac{\unit{J}}{\unit{kg}\cdot\unit{K}}
$$
at $0\,^\circ\unit{C}$ (degrees Celsius) \cite{SVK}.
%\footnote{http://www.efunda.com/Materials/common\_...}
%varying up to $\approx  805 \frac{J}{Kg\cdot K}$ at $500$ degrees Celsius.
Computing the dimensionless quantity
$\Gamma=\bar R/c_v$, we thus obtain 
$$
\Gamma \approx 0.401,
$$
in remarkable agreement with the value $\Gamma=.4$ predicted
by statistical mechanical approximation $\Gamma=\frac{2}{2n+1}$
for a diatomic gas, $n=2$.
(Recall that air is composed of roughly $78\%$ nitrogen,
$21\%$ oxygen, and $1\%$ neon, so it is essentially a diatomic mixture.)

For thermal conductivity, the ratio of heat flux to temperature
gradient asserted by Fourier's law of heat conduction, 
we have the value
$$
\kappa \approx .025\, \frac{\unit{W}}{\unit{m}\cdot\unit{K}},
$$
$\unit{W}=$ Watts $=$ Joules per second,
%\footnote{http://en.wikipedia.or/wiki/Thermal\_conduc...}
and for dynamic or ``first'' viscosity, the rate of
proportionality of shear stress to velocity gradient
of a shear flow asserted in Newton's law of viscosity,
the value
$$
\mu_1\approx
(1.78)\times 10^{-5}\, \frac{\unit{kg}}{\unit{m}\cdot \unit{s}}.
$$
$\unit{m}=$ meters, $\unit{s}=$ seconds \cite{Bro2,Whi}.
%\footnote{http://en.wikipedia.org/wiki/Viscosity...}
%
Collecting these values, we may compute the ``constant-volume Prandtl 
number''
\begin{equation}\label{prandtl}
{\rm Pr}_v:= c_v \mu_1/\kappa \approx
\frac{ (716)(1.78)\times 10^{-5}}{ (.025)}
\approx .510,
\end{equation}
or very nearly $1/2$.
This is related to the usual (constant-pressure) Prandtl number
${\rm Pr}:=c_p \mu_1/\kappa$ by ${\rm Pr}_v={\rm Pr}/\gamma$, where
$\gamma:=c_p/c_v$ is the heat capacity ratio, or {\it adiabatic index},
relating specific heat at constant pressure $c_p$ to specific heat at
constant volume $c_v$, under ideal gas assumptions, 
$\gamma=1+\Gamma$.  For a variety of gases over a fairly wide range of temperatures \cite[p. 43]{Whi},
${\rm Pr}\approx 0.7$.
%\footnote{http://en.wikipedia.org/wiki/Prandtl\_number;
%http://en.wikipedia.org/wiki/Adiabatic\_index}

Recall that the effective viscosity appearing in the
one-dimensional Navier--Stokes equations is $\mu=2\mu_1+\mu_2$,
the sum of twice the dynamic viscosity and the ``second viscosity'' $\mu_2$, 
which is commonly taken as $\mu_2=-(2/3)\mu_1$ based on the assumption
that pressure equals ``mean pressure'' defined as one-third the
trace of the three-dimensional stress tensor, 
an approximation that appears to be in good agreement with
experiment at least for monatomic and diatomic gases
\cite{Ba, Ro}.
We therefore take $\mu\approx (4/3)\mu_1$, 
giving
\begin{equation}\label{prandtlel}
\nu/\mu = (3/4){\rm Pr}_v^{-1} \approx 1.47.
\end{equation}
Interesting values for computation are thus
$\Gamma=.4$ ($\gamma=1.4$), $\nu/\mu=1.47$, modeling air.

\begin{remark}\label{alter}
A convenient alternative formula involving commonly tabulated
dimensionless quantities is
\begin{equation}\label{prandtlrel2}
\nu/\mu = (3/4)\gamma{\rm Pr}^{-1},
\end{equation}
where ${\rm Pr}$ denotes (usual) Prandtl number and $\gamma=c_p/c_v$
the adiabatic index.
Assuming the value ${\rm Pr}\approx 0.7$ typical of simple
(e.g., monatomic, diatomic) gases, we obtain the
general rule of thumb
\begin{equation}\label{simplerule}
\nu/\mu \approx (1.09)\gamma;
\end{equation}
see Table \ref{Prcomp} for more precise values.
%valid for arbitrary gases obeying the ideal gas assumptions.
\end{remark}

\begin{remark}\label{altGrun}
We note also the alternative formula
$$
\Gamma= \Big(\frac{c_p}{\bar R}-1\Big)^{-1}
$$
for the Gruneisen constant 
in terms of specific heat at constant
pressure. (For air, $c_p\approx 1\,\unit{J}/\unit{g}\cdot \unit{K}= 1,000\, \unit{J}/\unit{kg}\cdot\unit{K}$.)
%TODO, CHECK physical tables on web: $c_p\approx 1.0$ I seem
%to recall, so $c_r\approx 1.43$, ${\rm Pr}\approx .715$.
\end{remark}
%TODO: what is standard notation for specific heat ratio???

\section{Gas constants for other gases}\label{simple}
\subsection{Ideal Gases}
Using the dimensionless formula 
\eqref{prandtlrel2}, we next investigate typical parameter values
for some common gases.
For an ideal gas, the Prandtl number is predicted
by {\it Eucken's formula} \cite{Bro}
\begin{equation}\label{Euc}
{\rm Pr}=\frac{4\gamma}{9\gamma-5},
\end{equation}
giving in combination with \eqref{prandtlrel2} the simple relation
\begin{equation}\label{simplest}
\nu/\mu= (3/4)\frac{9\gamma -5}{4}.
\end{equation}
We display in Table \ref{Prcomp} 
the relation for various gases between the theoretical
value \eqref{Euc} and experimentally measured values for Pr
(Table 1.9-1 of \cite{Bro}, 
adapted from \cite{Lo}, p. 250).
Though not displayed, experimental values of the adiabatic index
$\gamma=c_p/c_v$ match
theoretical predictions to within $1\%$ for monatomic and diatomic
gases, $3\%$ for triatomic, and $4.8\%$ for five-atomic gas CH$_4$,
as do experimental values vs. theoretical predictions of
the Gruneisen coefficient $\Gamma=\bar R/c_v$.
%NOTE: the experimental $\gamma$ match theoretical values to within
%$3\%$ for triatomic, $1\%$ for monatomic and diatomic,
%$4.8\%$ for CH$_4$, $5$-atomic.

In summary, the equation of state and temperature law
predicted by ideal gas theory appear to be extremely accurate
at usual temperatures,
while the predictions involving viscosity ($\mu$, Pr, $\kappa$)
degrade with molecular complexity, being nearly exact for
monatomic gases, quite good for most diatomic gases, but
only a first approximation for triatomic and more complicated
gases.
Indeed, the derivation of viscosity and heat conduction 
formulae involves additional simplifying assumptions
whose validity degrades with structural complexity
\cite{Bro}.
{\it Thus, we may use with some confidence
the theoretical prediction \eqref{simplest}
for simple gases,
but must consult experimental data 
for complex gases}.
% or liquids.}  NO!

\begin{table}
\begin{center}
\begin{tabular}{| c | c | c | c | c | c | c|}
\hline
Gas & $\gamma$ (th.) & Pr (th.) & Pr (exp.) & $\nu/\mu$ (th.) 
& $\nu/\mu$ (exp.) & rel. err.\\
\hline
He &     5/3  &   .667 &    .694  & 1.88  & 1.80 & 4\%\\
Ar &     &   &.669   & & 1.87& .05\%\\
${\rm H}_2$ &  7/5     &   .737 &     .712  & 1.43  & 1.47 & 2\%\\
${\rm N}_2$ &       &   &  .735  &  & 1.43& 0\%\\
${\rm O}_2$ &       &   & .732  &  & 1.43& 0\%\\
CO          &       &   &  .763  &  & 1.38& 3.5\%\\
CO$_2$    &   4/3 & .762   &  .819 &  1.31 & 1.22& 6.8\%\\
H$_2$S    &    &  &  .929 &  & 1.08 & 17.6\%\\
%exp. heat ratio. 1.25
%SO$_2$    &    & &  .833 &  & 1.20 & 8\%\\
SO$_2$    &    & &  .833 &  & 1.17 & 10.7\%\\
CH$_4$    &   5/4 & .8   &  .777 & 1.17  & 1.28 & 9.5\%\\
\hline
\end{tabular}
\bigskip
\caption{
Theoretical vs. experimental values of ${\rm Pr}$ and $\nu/\mu$
and relative error in $\nu/\mu$ at $20\,^\circ\unit{C}=68\,^\circ\unit{F}$ (room 
temperature).}
%TODO: how make degrees??? (here and elsewhere-search `degree') DONE-gdl
\label{Prcomp}
\end{center}
\end{table}

\begin{remark}
From \eqref{gammaformula} and \eqref{Euc} 
we obtain the theoretical range
$$
\gamma\in [1,1.66...],\quad
\nu/\mu \in [.75, 1.88].
$$
%well within range of our computations.
\end{remark}

\subsection{Temperature dependence and the 
kinetic theory of gases}\label{kinetic}
%TODO: more directed references here...
Though the ratio \eqref{Euc} of viscosity and heat conductivity predicted by
the kinetic theory of gases is constant, the absolute size predicted
for either one depends on temperature, $T$.
For example, the predicted viscosity for a monatomic gas
obtained through Chapman--Enskog expansion
of the Boltzmann equations with hard-sphere potential is 
{\it Chapman's formula}
\begin{equation}\label{chapman}
\mu_1= \mu_1(T)=(2/3)
\sqrt \frac{mkT}{\pi \sigma^2},
\end{equation}
where $m$ is mass per particle in kilograms $\unit{kg}$, $k$ is the Boltzmann
constant in Joules per degrees Kelvin $\unit{J}/\unit{K}$, $T$ is temperature in 
degrees Kelvin $\unit{K}$, and $\sigma$ is the collision cross section
in meters squared $\unit{m}^2$, given approximately by one-half diameter
squared \cite{Ba,Bro}.
This appears to give good physical agreement, and a
refined version given by
{\it Sutherland's formula}
\begin{equation}\label{sutherland}
\mu_1= \frac{(1+\tilde m)T^{3/2}}{T+\tilde m},
\end{equation}
$\tilde m$ constant, to give extraordinarily good agreement \cite{Bro}.
More generally, viscosity is typically modeled by
\begin{equation}\label{viscdep}
\mu_1 =CT^q,
\quad
1/2\le q\le 1,
\end{equation}
with $q\approx .76$ for air \cite{L}.

%TODO: something like this?
%Properly, we should include the above-described temperature-dependence
%in the physical study of large-amplitude shock layers.
%This can readily be accomoated in our analysis, as described in
%Appendix \ref{tempdep}, with only minor changes.

%\begin{remark}\label{temprmk}
Properly, we should include the above-described temperature-dependence
in the physical study of large-amplitude shock layers.
Though beyond the scope of the present project, this appears to
be feasible by a slight modification of the same techniques,
as we discuss further in Appendix \ref{tempdep}.
%\end{remark}

\section{Liquids and dense gases}\label{dense}

%For dense gases or liquids, the ideal gas assumptions break down,
%and the polytropic equation of state \eqref{Gammaeq} is replaced
%by more sophisticated versions such as
%Peng--Robinson or ``stiffened polytropic'' equations of state
%\cite{H,HVPM}.

For comparison, values of 
the Prandtl number ${\rm Pr}$ for various media at $20\,^\circ\unit{C}$ are \cite[p. 80]{Whi}:
%\footnote{http://en.wikipedia.org/wiki/Prandtl\_number}
\begin{itemize}
\item around $0.024$ for mercury,
\item around $0.7$ for air and helium,
\item around $7$ for water,
\item around $16$ for ethyl alcohol,
\item around $10,000$ for castor oil, and
\item around $12,000$ for glycerin.
%\item around $7\times10^{21}$ for Earth's mantle
%\item between $100$ and $40,000$ for engine oil,
%\item between $4$ and $5$ for R-12 refrigerant
\end{itemize}
%TODO: Find specific heat ratios $\gamma:=c_p/c_v$ and
%Grueisen numbers $\Gamma:=R/c_v$, to determine 
%parameters $\Gamma$ and $\nu/\mu$.
%There is the second question whether ideal pressure law gives
%good model in these contexts-- check on web.
The adiabatic index (specific heat ratio) of water is 
$\gamma\approx 1.01\approx 1.0$,
%with $C_v$ ($c_v$ expressed in moles) of $\approx 74.53$.
%Thus, 
hence
$$
%\Gamma=R/C_v
%\approx \frac{287.05}{74.53} \approx 3.85,
%NO! THIS IS SPECIFIC GAS CONST FOR AIR! WRONG CALC!
%\quad
\nu/\mu= \gamma/ {\rm Pr}\approx {\rm Pr}^{-1}\approx .143,
$$
quite far from the gas values of Table \ref{Prcomp}.
%with the
%high value of $\Gamma$ indicating the essentially incompressible 
%nature of the fluid.
%I'll look up phenomenological setting of $\Gamma$ for water, and
%see what we can do...
%(though anyway, it's rather incompressible as we see by high $\Gamma$... -K)
%

Moreover, for dense gases or liquids, the ideal gas assumptions break down,
and the polytropic equation of state \eqref{Gammaeq} must be replaced
by more sophisticated versions such as
Peng--Robinson or ``stiffened polytropic'' equations of state
\cite{H,HVPM}.
%
%Behavior does not match the ideal gas equation of state
%$p=\Gamma \rho e$, but is better modeled 
For example, water is often modeled
by a stiffened equation of state
\begin{equation}\label{stiffenedeos}
p=\Gamma \rho e - \gamma P_0
\end{equation}
behaving like a prestressed material, with base stress $P_0$
and $\Gamma$ determined empirically: for example, $\gamma \sim 6.1$ and
$$
P_0=2,000\, \unit{MPa}=2\times 10^9\, \unit{N}/\unit{m}^2= 
%2GPa
%20,000 atmospheres
2\times 10^9\,\unit{kg} /\unit{m} \cdot \unit{s}^2
$$
or $\gamma=7.42$ and $P_0=296.2\, \unit{MPa}$ \cite{IT,HVPM}.
Similar techniques are used to model liquid argon, nickel, mercury, etc.
\cite{H,CDM}.

It would be very interesting to investigate the effects on stability
of these modifications in the polytropic equation of state.
For the moment, what we can say, physically, 
is that {\it insofar as they conform
to a polytropic (ideal gas) equation of state, shock waves are stable}.
However, above a critical density, even standard, simple (e.g., monatomic
or diatomic) gases
are observed {\it not} to conform to a polytropic law \cite{C},
and in this regime our mathematical conclusions hold no sway.

\section{Computation of the Mach number}\label{mach}

A dimensionless quantity measuring shock strength is 
the {\it Mach number} 
\[
M = \frac{u_- - \sigma}{c_-}
\]
(for a left-moving shock), where 
$u_-$ is the downwind velocity, $\sigma$ is the shock speed, 
and $c_-$ is the downwind sound speed, all in Eulerian coordinates.  
The conservation of mass equation in Eulerian coordinates
is $\rho_t + (\rho u)_x = 0$,  giving jump condition 
$\sigma [\rho] = [\rho u]$, or, in Lagrangian coordinates,
\[
\sigma = \frac{u_+ v_- - u_- v_+}{v_- - v_+}.
\]
Thus, 
\[
M = \frac{u_- - \sigma}{c_-} = \frac{v_-(u_- - u_+)}{c_-(v_- - v_+)} = \frac{v_- [u]}{c_- [v]} = -s \frac{v_-}{c_-},
\]
which, under our normalization $s=-1$, $v_-=1$, gives
$$
M=c_-^{-1}=(\Gamma (1+\Gamma) e_-)^{-1/2}
$$
or, using $e_-=\frac{(\Gamma +2)(v_+-v_*)}{2\Gamma(\Gamma +1)}$,
\begin{equation}\label{machform}
M=\frac{\sqrt{2}}{\sqrt{(\Gamma +2)(v_+-v_*)}}.
\end{equation}
In the strong-shock limit $v_+\to v_*$, $M\sim (v_+-v_*)^{-1/2}$;
in the weak shock limit $v_+\to 1$,  $M \to 1$.

\section{Lifted Matrix bounds}\label{liftbd}

%Finally, we
We establish the following useful bound on the operator
norm of the ``lifted'' matrix $A^{(k)}$ acting on $k$-exterior products
$V=V_1\wedge \cdots \wedge V_k$ by the operation
$$
A^{(k)}V:=
AV_1\wedge \cdots \wedge V_k+ \dots + V_1\wedge \cdots \wedge AV_k.
$$
induced by a given matrix $A$, where $A^{(k)}$ by convention is
represented with respect to standard basis elements
$e_{j_1} \wedge \cdots e_{j_k}$.

\begin{lemma}\label{Akbd}
In the matrix operator norm $|\cdot |_p$ with respect to $\ell^p$,
\begin{equation}\label{Akbdeq}
|A^{(k)}|_p\le k|A|_p,
\quad 1\le p\le \infty. 
\end{equation}
\end{lemma}

\begin{proof}
It is sufficient to establish \eqref{Akbdeq} for $p=1$ and $p=\infty$,
the result for other $p$ following from the Riesz--Thorin interpolation
theorem
%TODO: reference here please.
$$
|M|_p\le |M|_{s_1}^{\theta_1} |M|_{s_2}^{\theta_2},
$$
$\theta_j>0$, $\theta_1 + \theta_2=1$, for $s_1<p<s_2$.

The $\ell^1$ operator norm is equal to the maximum over all columns
of the sum of moduli of column entries, or, equivalently, the maximum
$\ell^1$ norm of the image of a standard basis element.
Applying this definition to $A^{(k})$, we find readily that
the $\ell^1$ norm of the image of a standard basis element
is bounded by the sum of the $\ell^1$ norms of $k$ terms of form 
$$
Ae_{j_1}\wedge \cdots\wedge e_{j_k},
$$
expanded in standard exterior product basis elements.
As each of these clearly has $\ell^1$ norm bounded by
the $\ell^1$ norm of $Ae_{j_1}$, and thus by $|A|_{1}$,
we obtain the result for $p=1$.
The result for $p=\infty$ may be obtained by duality,
using $(A^{(k)})^*=(A^*)^{(k)}$ together with
$|M|_\infty=|M^*|_1$.
\end{proof}

\section{Temperature-dependent viscosity}\label{tempdep}

Finally, we discuss the changes needed to accommodate a common
temperature or other dependence in the
coefficients of viscosity and heat conduction.
For concreteness, we focus on the case
\begin{equation}\label{visclaw}
\mu(e)=\mu_* e^q, \quad
\kappa (e)=\nu_* e^q,
\qquad
%TODO: for theoretical treatment (later), might change to
%0<q < 1,
0\le q \le  1,
\end{equation}
$\mu_*$, $\nu_*$ constant,
encompassing the Chapman formula 
\eqref{chapman} predicted by the kinetic theory of gases
as well as the more general formula \eqref{viscdep},
indicating afterward by a few brief remarks the extension
to more general situations.

\subsection{Rescaling}\label{trescale}
Under \eqref{eq:ideal_gas}, \eqref{visclaw}, 
it is easily checked that the Navier--Stokes
equations are invariant under the modified rescaling 
\begin{equation}\label{tscaling}
(x,t,v,u,T)\to 
(-\epsilon s|\epsilon s|^{-2q}x, \epsilon s^2|\epsilon s|^{-2q}t,
 v/\epsilon, -u/(\epsilon  s), T/(\epsilon^2 s^2)),
\end{equation}
consisting of \eqref{scaling} with $x$ and $t$ rescaled by the
common factor $|\epsilon s|^{-2q}$.
For the Chapman viscosity $q=1/2$, this reduces to
\begin{equation}\label{chapscaling}
(x,t,v,u,T)\to 
(-(\sgn s) x, |s|t,
 v/\epsilon, -u/(\epsilon  s), T/(\epsilon^2 s^2)),
\end{equation}
essentially undoing the rescaling in the $x$-coordinate.
Evidently, the Rankine-Hugoniot analysis of Section \ref{RH}
goes through unchanged. 

\subsection{Profile equations}\label{tprof}
The profile equations \eqref{eq:ideal_profile1}--\eqref{eq:ideal_profile2}
are unaffected by dependence of $\mu$, $\kappa$.
Setting $\nu_*:=\kappa_*/c_v$, and making the change of independent
variable
\begin{equation}\label{yvar}
\frac{\dif x}{\dif y}= \mu= \mu_* e^q,
\end{equation}
we may thus reduce them to the form 
\begin{align}
v'&=\frac{1}{\mu_*}\left[v(v-1)+\Gamma (e-ve_-)\right],\label{eq:tprofile1}\\
%TODO: fix! parentheses in wrong place! (DONE).
e'&=\frac{v}{\nu_*}\left[-\frac{(v-1)^2}{2}+(e-e_-)+(v-1)\Gamma e_-\right]
\label{eq:tprofile2}
\end{align}
already treated in the constant-viscosity case, $'=\frac{d}{dy}$.
To obtain the profile in original $x$-coordinates, we have
only to recover the change of coordinates $x=x(y)$ by
solving \eqref{yvar} with $e=\hat e(y)$, $\hat e$ the constant-viscosity
profile.

Recalling that $\hat e(y)=e_\pm + O(e^{-\theta |y|})$,
$\theta>0$ for $y\gtrless 0$, where $e_+>0$ and $e_->0$
except in the strong-shock limit $v_+\to v_*$,
we find that $x$ and $y$ are equivalent coordinates on $x>0$,
and on $x<0$ are equivalent for $v_+$ bounded from the strong-shock limit $v_*$.
However, in the strong-shock limit $v_+=v_*$, $e_-=0$, we have
the interesting phenomenon that the $x<0$ part of the shock profile
terminates at a finite value $X_-=x(-\infty)$, since
$$
|x(-\infty)|=|\int_0^{-\infty} (e(y)/c_v)^q\,\dif y|
\le C\int_0^{-\infty} e^{-q\theta |y|}\,\dif y <+\infty.
$$

\begin{remark}
Holding $(v,u,e)_-$ fixed, and varying $v_+$ toward its
minimum value $v_{*}$,
%NOTE: v unchanged by rescaling, so long as v_-=1..
we find that 
$$
s=\sqrt{-[p]/[v]}\to \infty
$$
as $p_+\sim e_+ \sim (e_+/e_-)\sim (v_+-v_*)^{-1}$, since
$v_+$ is bounded from zero, $e_-$ is being held constant,
and ratio $e_+/e_-$ is independent of scaling so may be
computed from formulae \eqref{e-}--\eqref{e+}. 
Thus, shock width in the constant-viscosity case is 
of order $|v_+-v_*|$ going to zero as shock strength
(measured in specific volume $v$) goes to 
its maximum value of $|1-v_*|$.
By comparison, for the Chapman viscosity $\mu\sim e^{1/2}$,
the shock width remains approximately constant in the strong-shock
 limit, a significant deviation in the theories.
See, e.g., \cite{Tru,H} for further discussion of this and related points.
\end{remark}

\begin{remark}\label{anydep}
Clearly, the same procedure may be used to determine profiles
for arbitrary $\mu=\mu(v,u,e)$, $\kappa/\mu$ constant,
setting $\frac{\dif x}{\dif y}=\mu(v,u,e)$ in \eqref{yvar}.
\end{remark}

\subsubsection{Limiting behavior}\label{limprof}
The limiting profile equations at $v_+=v_*$ are
\begin{align}
v'&=\frac{1}{\mu_*}\left[v(v-1)+\Gamma e\right],\label{eq:tlim_profile1}\\
e'&=\frac{v}{\nu_*}\left[-\frac{(v-1)^2}{2}+e\right].
\label{eq:tlim_profile2}
\end{align}
Linearizing about $U_-$ gives
\begin{equation}\label{tlinprof}
\begin{pmatrix}
v \\ e
\end{pmatrix}'=
M
\begin{pmatrix}
v \\ e
\end{pmatrix},
\qquad
M:=
\begin{pmatrix}
\mu_*^{-1} & \Gamma\mu_*^{-1}\\
0  & \nu_*^{-1}\\
\end{pmatrix}.
\end{equation}
For $\nu_*/\mu_*>1$, we have, therefore, that the
slow unstable manifold at $-\infty$ is tangent to
$(\Gamma (\nu_*/\mu_*)/(1-(\nu_*/\mu_*)),1)$, with growth rate
$\sim e^{-\nu_*^{-1}y}$, and thus generically
\begin{equation}\label{newrel}
\frac{\hat e_y}{\hat e}\sim  
%\nu_*^{-1} ,
\frac{1}{\nu_*},
\quad
\frac{\hat u_y}{\hat e}=
\frac{\hat v_y}{\hat e}\sim 
%-\nu_*^{-1}\Gamma/(1-\nu_*^{-1}) \quad
%\frac{\Gamma}{1-\nu_*} \quad
\frac{\Gamma /\mu_*}{1-\nu_*/\mu_*} \quad
\hbox{\rm as $y\to -\infty$}.
\end{equation}
For $\nu_*/\mu_*<1$, the slow manifold is
tangent to $(1,0)$, and we have the opposite situation that, generically,
$\frac{\hat e_y}{\hat e}\sim  \mu_*^{-1}$, 
$\frac{\hat u_y}{\hat e}\to \pm \infty$ as $y\to -\infty$.

\subsection{Eigenvalue equations}\label{teig}
Computing the linearized integrated eigenvalue equations as in 
Section \ref{linint}, and making the change of variables \eqref{yvar},
we obtain after some rearrangement, the modified, temperature-dependent
equations
\begin{equation}\label{yevalue}
\begin{aligned}
&\lambda \hat \mu v+v'-u'=0,\\
&\lambda \hat \mu u
+\Big[ 1+ \frac{q\hat e_y}{\hat e\hat v} 
-\frac{q\hat u_y^2}{\hat \mu \hat e \hat v} \Big]u'
+\Big[\frac{\Gamma}{\hat v} 
-\frac{q \hat u_y}{\hat e \hat v} \Big]\eps'
+\frac{\Gamma \hat u_y}{\hat v}u +\left[
-\frac{\Gamma \hat e}{\hat v^2}+
\frac{\hat u_{y}}{\hat v^2}\right]v'
=\frac{ u''}{\hat v},\\
&\lambda \hat \mu \eps+
%WOW! Cancellation!
%\Big[ 1 +\frac{\nu_* q\hat e_y}{\hat e \hat v}
%-\frac{q\nu_* \hat e_y}{\hat e \hat v} \Big]\eps'
\eps'
+ \left[\hat u_y-\frac{\nu_* \hat u_{yy}}{\hat v} \right] u
+\left[\frac{\Gamma\hat e}{\hat v}-(\nu_*+1)\frac{\hat u_y}{\hat v} 
-\frac{q\nu_* \hat u_y\hat e_y}{\hat \mu \hat e \hat v} 
\right]u'\\
&
\qquad \qquad \qquad \qquad \qquad
\qquad \qquad \qquad \qquad \qquad
\qquad \qquad 
+\left[\frac{\nu_*\hat e_y}{\hat v^2}\right]v'
=\frac{\nu_*}{\hat v}\eps'',
\end{aligned}
\end{equation}
$'=\frac{\dif}{\dif y}$,
where $(\hat v, \hat u, \hat e)=(\hat v, \hat u, \hat e)(y)$
are just as in the constant-viscosity case.

This yields a first-order eigenvalue system
\begin{equation}\label{tfirstorder}
W' = A(y,\lambda) W,
	\end{equation}
\begin{equation} \label{tevans_ode}
A(y,\lambda) =
\begin{pmatrix}
0 & 1 & 0 & 0 & 0 \\
\lambda \hat \mu \nu_*^{-1}\hat v & \nu_*^{-1}\hat v & \nu_*^{-1}\hat  v\hat u_y-\hat  
u_{yy} & \lambda \hat \mu  g(\hat U) & g(\hat U)- h(\hat U) \\
0 & 0&  0 & \lambda \hat \mu  & 1\\
0& 0& 0 & 0 & 1\\
0 & \Gamma-\frac{q\hat u_y}{\hat e} & \lambda \hat \mu \hat v+\Gamma\hat u_y 
& 
%\lambda \hat \mu 
%j(\hat U) 
%\hat v \big( 1 -\frac{q\hat u_y^2}{\hat \mu \hat e \hat v} \big)
\lambda \big( \hat \mu \hat v  -\frac{q\hat u_y^2}{\hat e}\big)
& f(\hat U)- \lambda \hat \mu 
\end{pmatrix},
\end{equation}
\begin{equation}
W = (\eps ,\eps', u, v, v')^T,
\quad \prime  
= \frac{d}{dy},
\end{equation}
where 
\begin{equation}
\begin{aligned}
f(\hat U)&:=\frac{\hat u_y-\Gamma \hat e}{\hat v}+
\hat v \Big( 1+ \frac{q\hat e_y}{\hat e\hat v} 
-\frac{q\hat u_y^2}{\hat \mu \hat e \hat v} \Big)
\\
&=
2\hat v-1-\Gamma e_-
 +\Big(\frac{q\hat e_y}{\hat e} 
-\frac{q\hat u_y^2}{\hat \mu \hat e } \Big),
\end{aligned}
\label{eq:tf_eq}
\end{equation}
\begin{equation}
\begin{aligned}
g(\hat U)&:=\nu_*^{-1}\Big(\Gamma\hat e-(\nu_*+1)\hat u_y \Big)
-\frac{q \hat u_y\hat e_y}{\hat \mu \hat e \hat v} ,
\label{eq:tg_eq}
\end{aligned}
\end{equation}
\begin{equation}
h(\hat U):=-\frac{\hat e_y}{\hat v}= 
-\nu_*^{-1}\left(-\frac{(\hat v-1)^2}{2}+(\hat e-e_-)+(\hat v-1)\Gamma e_-\right),
\label{eq:th_eq}
\end{equation}
%\begin{equation}
%j(\hat U):= \hat v \Big[ 1+ \frac{q\hat e_y}{\hat e\hat v} 
%j(\hat U):= 
%\hat v \Big[ 1 -\frac{q\hat u_y^2}{\hat \mu \hat e \hat v} \Big],
%\label{eq:tj_eq}
%\end{equation}
\begin{equation}
\hat \mu:=\hat e^q,
\label{eq:te_eq}
\end{equation}
reducing for $q=0$ to that of the constant-viscosity case.
We omit the details, which are straightforward but tedious.

Noting that all terms not appearing in the constant-viscosity
case involve derivatives of the profile, hence vanish at $y=\pm \infty$
so long as $e_\pm$ (appearing in denominators) do not vanish,
we may conclude by the constant-viscosity analysis consistent splitting
away from the strong-shock limit $v_+\to v_*$, $e_-\to 0$.
We may thus construct an Evans function that is analytic in $\lambda$
and continuous in all parameters away from the strong-shock limit.

%Moreover, under assumption \eqref{odd}, the coefficient $A$
%is bounded and smooth down to the strong shock limit $v_+\to v_*$;
%likewise, a straightforward analysis similar to that of the constant-viscosity
%case verifies consistent splitting for $v_+>v_*$ and (by direct calculation)
%for the formal limiting system $v_+=v_*$.

\subsubsection{Limiting behavior}\label{eiglim}
Assume, as for the physical cases considered above, that
\begin{equation}\label{odd}
\nu/\mu\equiv \nu_*/\mu_* >1.
\end{equation}
Then, by the limiting analysis \eqref{newrel}, 
all terms in $A(y,\lambda)$ {\it remain bounded and well-defined 
in the strong-shock limit}.
Thus, we may hope for convergence as in the constant-viscosity case.
%TODO: restore later?
%as investigated numerically in Section \ref{tdepcase}.

On the other hand, new terms $\hat e_y/\hat e$, $\hat u_y/\hat e$ of order
$$
\frac{\hat e- e_-}{\hat e} = 1- \frac{e_-}{\hat e}
$$
exhibit singular behavior in the $v_+ \to v_*$, $e_-\to 0$ limit
reminiscent of that of the isentropic case \cite{HLZ}.
In particular,
since $e_-/\hat e $ approaches its limiting value $1$ as $y\to -\infty$
only as $|\hat e-e_-|/e_-$, this means that $|A-A_\pm|$ does not
decay at uniform exponential rate as $y\to -\infty$, but only
as $O(e_-^{-1}e^{-\theta |y|})$, $\theta>0$, so that the strong-shock
limit is a singular perturbation in the sense of \cite{PZ, HLZ} and
not a regular perturbation as in the constant-viscosity case.
Thus, 
%TODO: which?
%a proof of convergence 
an analysis of behavior
in the strong-shock limit is likely
to involve a delicate boundary-layer analysis similar to that of
\cite{HLZ} in the isentropic case.
This appears to be a very interesting direction for further investigation.
%TODO: nonetheless, try to support this conjectured conv. numerically-
%I do think it is so! -KZ

%and convergence
%in the strong shock limit on compact subsets of $\Re \lambda \ge 0$.
%%We omit the details, which will be given in a future work. 
%%focusing on this case.
%
%We can thus carry out the analysis of the strong shock limit,
%{\it in the favorable $y$-coordinates}, essentially as in
%the constant-viscosity case.
%What remains to be done is a high-frequency analysis giving uniform
%bounds on unstable eigenvalues in order to conclude as in the
%constant-viscosity case stability in the strong-shock limit.
%%We intend to treat this physically important case in a future work.

\begin{remark}\label{unexpected}
The appearance of condition \eqref{odd} is unexpected,
dividing into two cases the physical behavior in the strong-shock limit.
\end{remark}

\begin{remark}\label{prelimrmk}
Our numerical experiments, though still quite preliminary 
(restricted to diatomic gas,
$\Gamma= .4$, $\nu_*/\mu_*= 1.43$, $q=0.5$) so far
again
indicate unconditional stability, also in the temperature-dependent case. 
%but not convergence in the strong shock limit.
\end{remark}

\subsection{General dependence}\label{gendep}
We conclude by discussing briefly the case of
general, possibly inhomogeneous dependence of viscosity on
$(v,u,T)$.
The homogeneous case goes exactly as before, working in $y$-coordinates
and noting Remark \ref{anydep} and the discussion of Section \ref{teig}.

The inhomogeneous case is well-illustrated by
Sutherland's formula \eqref{sutherland}.
Fixing a left-hand state (the ``true'' strong-shock limit),
without loss of generality $v_-=1$,
and taking $v_+\to v_*$, 
rescale by \eqref{tscaling} with $q=1/2$ and 
$\epsilon=1$, $s=s(v_+)$. 
The result {\it in the rescaled coordinates} is
\begin{equation}\label{intdep}
\mu= (4/3)s^{-1} \frac{(1+\tilde m)(s^2T)^{3/2}}{(s^2T)+\tilde m} 
=(4/3)\frac{(1+\tilde m)T^{3/2}}{T+s^{-2} \tilde m},
\end{equation}
converging in the strong-shock limit $v_+\to v_*$,
$s\sim |v_+-v_*|^{-1}\to \infty$ to the homogeneous 
Chapman formula 
$$
\mu= (4/3)(1+\tilde m)T^{1/2}.
$$
Other inhomogeneous laws go similarly, converging in the strong-shock
limit in each case to an appropriate homogeneous law.
Thus, {\it inhomogeneous dependence poses no essential
new difficulty}.

%TODO: include/improve this rmk?  Interesting, but relevance
%not that clear, esp. since relation of numerical visc. to
%actual discrete behavior quite unclear, esp. for large amplitudes.
%So, conclusion: NO, OMIT. -K
%\begin{remark}\label{numvisc}
%An interesting direction might be to study stability
%of Euler equations with numerical viscosity, depending
%on all variable through the flux Jacobian.
%\end{remark}\label{numvisc}

%TODO: review for final check...
%\section{TODO}
%Remaining todos, ordered by priority:
%
%1. Big job 1(a): double number of $\lambda$ points for $\Gamma=.2$.
%
%3. Convergence study (like Figure 5) for temperature-dependent viscosity
%(high priority).
%
%4. Graph iso-Mach-lines (fix this).
%
%5. Table 2, high-freq. radius bounds. (OMITTED BY DECISION)
%
%6. $\nu>>\mu$, $\nu<<\mu$, just two computations, for appropriate
%sections (done using high-freq. convergence not tracking bounds).
%OMITTED, but could be restored... NO POINT I THINK --K
%
%7. continuation in $v_+$ of initializing subspaces (low priority,
%but eventually should be done.
%
%OTHER? (Can't think of any...!)

\end{document}